\documentclass[11pt]{article}
\usepackage[T1]{fontenc}
\usepackage[utf8]{inputenc}
\usepackage{graphicx,}
\usepackage{complexity}
\usepackage{amsfonts,amssymb,amsmath,amsthm,mathtools,wrapfig}
\usepackage{tabularx,placeins,float,comment}
\usepackage{fullpage}
\usepackage{palatino}
\usepackage{setspace}
\usepackage{lineno}
\usepackage{enumerate}
\usepackage{tikz}
\usetikzlibrary{decorations.markings, positioning, shadows, arrows.meta}
\newtheorem{example}{Example}[section]
\newtheorem{theorem}{Theorem}[section]
\newtheorem{proposition}{Proposition}[section]

\newtheorem{lemma}{Lemma}[section]
\newtheorem{claim}{Claim}[section]

\newtheorem{remark}{Remark}[section]

\newcommand{\calL}{\mathcal{L}}
\newcommand{\calD}{\mathcal{D}}
\newcommand{\calR}{\mathcal{R}}
\newcommand{\calJ}{\mathcal{J}}
\newcommand{\calH}{\mathcal{H}}

\newcommand{\reach}{{\sc Reach}}  
\newcommand{\reachmf}{\textrm{\reach}_{M,F}}
\newcommand{\DA}{\sf DA}
\newcommand{\BA}{\sf BA}
\newcommand{\DS}{\mathbf{DS}}
\newcommand{\U}{\sf U}

\usepackage{xcolor}
\usepackage{thmtools, thm-restate}
\usepackage[colorlinks=true,citecolor=blue,linkcolor=blue]{hyperref}

\usepackage{collect}
\usepackage{ifthen}

\usepackage{todonotes}
\newcounter{todocounter}

\begin{document}
\title{On the Reachability Problem on Monoid-Labelled Undirected Graphs}
\author{Nagashri Krishnakumar\thanks{Indian Institute of Technology Madras, Chennai, India} \and Harshil Mittal$^*$ \and Jayalal Sarma$^*$}
\date{\today}
\maketitle
\begin{abstract}

The labelled reachability problem for undirected graphs where the edges are labelled with elements from a monoid $M$ (more generally, other groupoids or magmas) is known to capture space-bounded complexity classes \ComplexityFont{L} and $\NL$. Given a labelled graph $G(V, E)$ (labelled with $\phi ~ \colon E \to M$), $s,t \in V$ and an accepting subset $F \subseteq M$, the problem asks to test whether there is a walk $P$ from $s$ to $t$ in $G$ where $\phi(P) \in F$. When the accepting element is also part of the input, the problem has been studied by Ramaswamy {\em et al} (2019) \cite{RSS19} for the case when the monoid is aperiodic and the case when the monoid is a group. 
Motivated by the success in designing space-bounded algorithms for the graph reachability problem in the undirected case, we study the labelled reachability problem when the accepting set is also fixed. This allows us to reveal finer complexity upper and lower bounds, and dichotomies for the problem using the structure of the underlying monoid and the accepting set. In fact, when the monoid $M$ is a group or belongs to the monoid pseudovariety $\DA$, the previous results imply a deterministic logspace algorithm for the undirected labelled reachability problem over the monoid $M$ for any finite accepting subset as well. We prove the following:
\begin{itemize}
    \item For any finite monoid $M$, there is a deterministic logspace algorithm for solving the above problem when the accepting element is the identity of $M$. In contrast, if the accepting element is an idempotent element, under suitable constraints, the problem is $\NL$-hard.
    \item Considering the more general case of the problem when the accepting set is an arbitrary subset of $M$, we show a deterministic logspace algorithm for a commutative monoid $M$ for all accepting sets $F \subseteq M$.
    \item Further, we prove that for any finite $\calL$-commutative or $\calR$-commutative union-of-groups monoid $M$, there is a deterministic logspace algorithm for all accepting sets $F\subseteq M$. We also design deterministic logspace algorithms for the union-of-groups case that are neither $\calL$-commutative nor $\calR$-commutative, under certain constraints.
    \item For the specific monoids $\BA_2$ and $\U$, which are the building blocks of non-commutative aperiodic monoids, we prove a dichotomy theorem with respect to the accepting subsets by showing that for all subsets $F \subseteq M$, the reachability problem is either $\NL$-complete or is in \ComplexityFont{L}.
    
\end{itemize}
To obtain our results, we critically exploit the relationships between Green's equivalence relations in the union-of-groups monoids and the properties of the product graph (a graph introduced by \cite{RSS19}). 
\end{abstract}

\tableofcontents

\section{Introduction}

One of the fundamental computational problems related to a groupoid $(\mathcal{A},*)$\footnote{also known as a magma, $\mathcal{A}$ is a set closed under a binary operation $*$} is that of a \textit{word problem} - given two input words over the alphabet $\mathcal{A}$, test whether they represent the same element. A simplified version of this is the following: given a word over the alphabet $\mathcal{A}$ and a fixed element $\alpha \in \mathcal{A}$, test if the given word represents the element $\alpha$ or not, with respect to the operation. The general form of this, even on richer algebraic structures such as groups, when the groups are finitely presented, is undecidable \cite{Boo58,Nov58}. However, when the group is finite and fixed, the problem is indeed efficiently solvable and has interesting connections to low-level complexity theory. For example, the word problem over any finite non-solvable group exactly captures the class of problems efficiently solvable by logarithmic-depth circuits with bounded fan-in $\land, \lor$ and $\lnot$ gates \cite{Bar86}. Several striking finer characterisations of circuit complexity classes using these problems are also known \cite{BT88,CFL83}.

On the other side of complexity theory, the problem of testing reachability in directed graphs is a fundamental algorithmic problem that encapsulates many challenges in space-bounded complexity theory. Given a graph $G(V, E)$ and two special vertices $s$ and $t$, the problem asks to test if there is a path from the vertex $s$ to $t$ in the graph $G$. The problem in its most general form (directed graphs) is complete for the class $\NL$ (the class of problems that are solvable by non-deterministic Turing machines using only $O(\log n)$ space), and for undirected graphs it is known to be complete \cite{Rei04} for the class \ComplexityFont{L} (the class of problems that are solvable by deterministic Turing machines that use only $O(\log n)$ space). The $\NL$ vs \ComplexityFont{L} problem thus amounts to the question of whether it is possible to design a deterministic algorithm that uses only $O(\log n)$ space to test reachability in directed graphs. Various algorithm design techniques have been employed to obtain $O(\log n)$ space-bounded deterministic algorithms for special cases of the directed reachability problem; see \cite{All07}.

Generalising the above two themes - word problem and the reachability problem - 
Ramaswamy {\em et al}~\cite{RSS19} studied the $\mathcal{A}$-\reach  ~problem. 
Fix a finite groupoid $(\mathcal{A},*)$. Consider a directed graph $G(V, E)$ with labelling function $\phi: E \rightarrow \mathcal{A}$. This labelling naturally extends to walks in the graph as well. Indeed, for a walk $P = (v_1, v_2, \ldots, v_k)$ in the graph $G$,  $\phi(P) = \prod_{i=1}^{k-1} \phi(v_i,v_{i+1})$  \footnote{When $(\mathcal{A},*)$ is non-associative, $\phi(P)$ can instead be defined as the set of all values that $\prod_{i=1}^{k-1} \phi(v_i,v_{i+1})$ takes as one varies over all possible parenthesisations; then, in the definition of $\mathcal{A}$-\textsc{Reach}, the condition $\phi(P)\in F$ changes to $\phi(P)\cap F\neq \{\}$.}
, which we call as the \textit{yield} of the walk $P$. The problem of $\mathcal{A}$-\textsc{Reach} is as follows: given a graph $G(V, E)$ with labelling function $\phi: E \rightarrow \mathcal{A}$, a pair of vertices $s,t \in V$, a non-empty subset $F \subseteq \mathcal{A}$, is there a \textit{walk} from $s$ to $t$ with $\phi(P) \in F$? Indeed, the word problem is a special case, where the graph is a directed path from a source vertex $s$ to the destination vertex $t$, with the labels on the edges being the given sequence of symbols in $\mathcal{A}^*$ in the same order, and $F= \{\alpha\}$.

Motivated by the fact that the reachability problem is known \cite{Rei04} to be space-efficiently solvable for undirected graphs, Ramaswamy \textit{et al}~\cite{RSS19} study the above problem restricted to undirected graphs. One of the main algebraic structures studied in \cite{RSS19} is monoids. In particular, they studied, for a fixed monoid $M$, the problem \reach$_M$, which, given a labelled undirected graph $G(V, E)$ (labelled with $\phi: E \to M$), $s,t \in V$ and an accepting non-empty subset $F \subseteq M$, asks to test whether there is a walk $P$ from $s$ to $t$ in $G$ where $\phi(P) \in F$. It was shown that for every monoid $M$ which is a group, \reach$_M$ is solvable in deterministic logspace. On the other hand, when the monoid $M$ is an aperiodic monoid\footnote{A monoid is said to be aperiodic if for every element $a \in M$, there exists a $k \in \mathbb{N}$ such that $a^k=a^{k+1}$. The number $k$ is called the index of the element $a$ in $M$.}, the following dichotomy was shown: \reach$_M$ is either in deterministic logspace or is $\NL$-complete. This dichotomy theorem uses the fact that the accepting subset is a part of the input.

Since there are many algorithmic tools~\cite{Rei04} to study reachability in undirected graphs, it is conceivable that complexity theoretic questions such as the $\NL$ vs \ComplexityFont{L} problem may be amenable to those techniques applied to the labelled generalisation of the undirected reachability problem. Given this, it is important to study the problem even for various accepting subsets, not just restricted to the dichotomy theorem stated above.\\[-3mm]

\noindent{\bf Our Results}:
In this paper, we study the complexity of the labelled reachability problem where the monoid $M$ and the final accepting set $F$ are fixed, and for different choices of $M$ and $F$. More precisely, we define the following problem, which we will call the \textit{monoid reachability problem}. Fix a finite monoid $M$ and a non-empty subset $F \subseteq M$, define:
$$\textrm{\reach}_{M,F} = \left\{ (G(V,E),\phi,s,t) : \begin{array}{l}
\textrm{$G$ is an undirected graph, $\phi : E \to M$}\\
\textrm{is a labelling function, and $\exists$ a walk $P$}\\
\textrm{from $s$ to $t$ in $G$ such that $\phi(P) \in F$}
\end{array} \right\}
$$

\noindent{\bf Reachability with Identity as the Accepting Element}:
In the first version of the problem, we fix $F$ to be the monoid identity. Notice that, when the monoid is a group, it follows from \cite{RSS19} that there is a deterministic logspace algorithm for this problem. 

\label{Logspace algo when accepting set is inside unit group for union-of-groups}

We begin by showing a logspace algorithm for any fixed finite monoid with the accepting set being the monoid identity (Section~\ref{subsec: Logspace Algorithm for all Finite Monoids}). 
\begin{restatable}[]{theorem}{monoididentityalgo} 
\label{thm: logspace algo for accepting identity for any finite monoid}    
Consider any fixed finite monoid $M$. Let $\tilde{G}$ denote the unit group of $M$. Then, $\reachmf$ can be solved using a deterministic logspace algorithm for all accepting sets $F\subseteq \tilde{G}$  (and so, in particular, when $F=\{id\}$, where $id$ denotes the identity of $M$).
\end{restatable}

The above theorem establishes deterministic logspace algorithms for all finite monoids when the accepting element is the identity by using the embedding of the monoid into the transformation monoid of appropriate order. 
However, we note that for special classes of monoids, we can obtain such algorithms using the specific structure of the monoids, using different arguments, which may be of independent interest. We include them in Section \ref{subsec: Logspace Algorithm for all Finite Monoids}.

Complementing the above, we show that there are non-commutative monoids where the problem is $\NL$-hard even when the accepting element is an idempotent element (notice that the identity element is always an idempotent element) (Section~\ref{subsec: NL-hardness with an Idempotent as the Accepting Element}). An absorber in a monoid is an element $a \in M$ such that for every $b$, $ba = a$ and $ab = a$.

\begin{restatable}[]{theorem}{specialidempotent}
\label{idempotent hardness with certain conditions}
Consider any fixed finite monoid $M$ with absorber $0$. Let $\mu\neq 0$ be any idempotent element of $M$ that can be factored as $\mu = a\cdot b$ for some $a, b\in M$ such that one of the following holds true:\\[-8mm]
\begin{enumerate}[(i)]
\itemsep -4pt
    \item $b^2=0$ and $a$ is idempotent.
    \item $b^2=a^2=0$.
    \item $a^2=0$ and $b$ is idempotent.
\end{enumerate}
\vspace{-2mm}
Then, $\reachmf$ with $F=\{\mu\}$ is $\NL$-hard.
\end{restatable}

\noindent Note that there are explicit monoids (example, $\BA_2$, $\U$) where the above property is true for an idempotent element.\\[-3mm]

\noindent{\bf Reachability with any subset as the Accepting Subset}:
We now turn to the most general case, where the accepting set is an arbitrary $F \subseteq M$. It follows easily from the results of \cite{RSS19} that there are deterministic algorithms for any accepting subset $F \subseteq M$ for the problem $\reachmf$ in the cases when (a) the monoid $M$ is a group, or (b) when the monoid is in the class $\DA$ (which is a monoid pseudovariety contained in the class of aperiodic monoids). While there are $M$ and $F$ for which the problem $\reachmf$ is $\NL$-hard, we extend our understanding to several classes of monoids.

We begin with commutative monoids. We show a logspace algorithm for an arbitrary accepting set $F \subseteq M$, when the monoid $M$ is commutative (Section~\ref{subsec: Logspace Algorithms for Commutative Monoids}). 

\begin{restatable}[]{theorem}{commutativeArbitraryAcceptance} 
\label{thm: logspace algo commutative monoids}
Consider any fixed finite commutative monoid $M$. Then, $\reachmf$ can be solved using a deterministic logspace algorithm for all accepting sets $F\subseteq M$.    
\end{restatable}

Next, we generalise the case of groups to the union-of-groups.
\noindent We describe logspace algorithms for an arbitrary accepting subset $F \subseteq M$, when the monoid $M$ has additional properties as a union-of-groups. A monoid $M$ is said to be $\calR$-commutative (respectively $\calL$-commutative) if $a\cdot b$ and $b\cdot a$ belong to the same $\mathcal{R}$-class (respectively $\calL$-class) for all $a,b\in M$ (Section~\ref{subsubsec: Over L- (or R-) Commutative Union-of-Groups Monoids}).


\begin{restatable}[]{theorem}{firstlcommutativealgo} 
\label{Logspace algo when R-classes = H-classes}
Consider any fixed finite $\calL$-commutative or $\calR$-commutative union-of-groups monoid $M$. Then, $\reachmf$ can be solved using a deterministic logspace algorithm for all accepting sets $F\subseteq M$.
\end{restatable}

We remark that the above class also contains the class of Clifford monoids. For union-of-groups monoids, the property of $\calL$-commutativity can also be equivalently stated as the fact that $\calR$-classes and $\calH$-classes coincide. Motivated by specific union-of-groups monoids which are neither $\calL$-commutative nor $\calR$-commutative, we extend the above study further to relax the $\calL$-commutativity.

For example, consider the graph $2P_2$ (i.e., disjoint union of two paths, each on two vertices). Its endomorphism monoid $End(2P_2)$ consists of all homomorphisms from $2P_2$ to itself, with function composition as the monoid operation. Note that $End(2P_2)$ is a union-of-groups monoid\footnote{More generally, \cite{PKG16} showed that endomorphism monoids of complements of even cycles are union-of-groups.}, and its $\mathcal{R}$-classes and $\mathcal{H}$-classes do not coincide 
(see Figure~\ref{End(2P2) figure}).

\begin{figure}[H]
\centering
\includegraphics[scale=0.7]{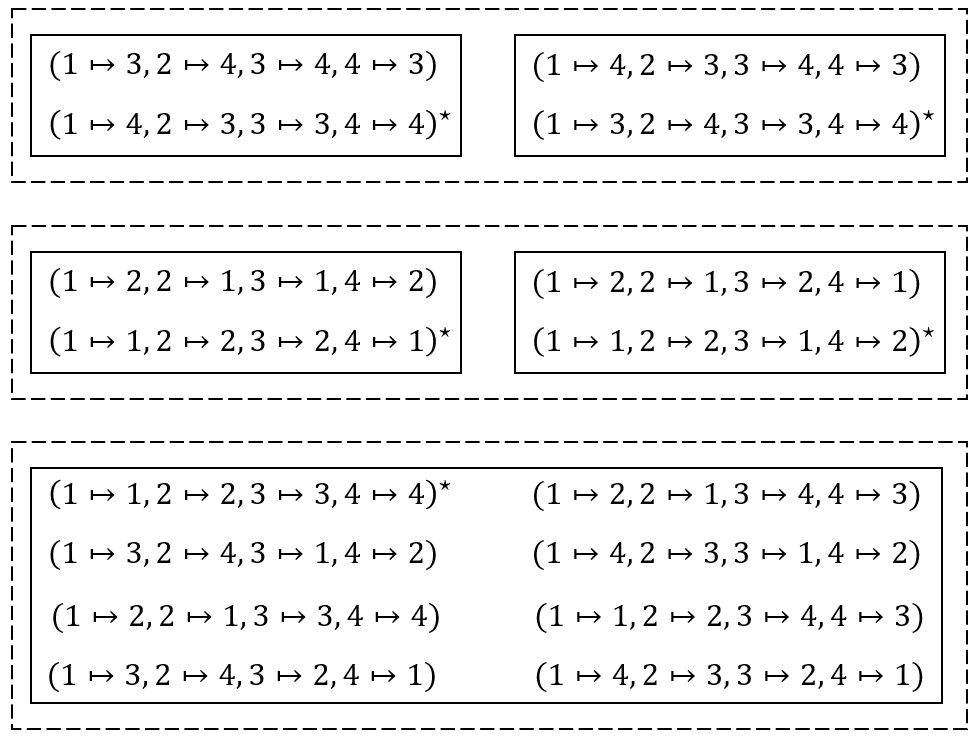}
\caption{This figure shows $\mathcal{R}$-classes and $\mathcal{H}$-classes of the union-of-groups monoid $End(2P_2)$ marked with dashed and solid boxes respectively. The idempotent elements are marked with $*$. The bottom box is the unit group. Note that $\mathcal{R}$-classes do not coincide with $\mathcal{H}$-classes and so, $End(2P_2)$ is not $\calL$-commutative. Also, observe that the two $\mathcal{R}$-classes outside the unit group are right ideals.}
\label{End(2P2) figure}
\end{figure}

However, in this case, the two $\mathcal{R}$-classes outside the unit group are right ideals. We prove a version of the previous theorem to handle this example as well when the accepting set $F$ is either of these two $\mathcal{R}$-classes (Section~\ref{subsubsec:When R-classes other than Unit Group are Right Ideals}). More generally:
\begin{restatable}[]{theorem}{firstrightidealalgo}
\label{R-classes right ideal thm} Consider any fixed finite union-of-groups monoid $M$ such that all $\mathcal{R}$-classes other than the unit group (say $\tilde{G}$)\footnote{The $\mathcal{R}$-class containing unit group $\tilde{G}$ has only one $\mathcal{H}$-class, namely $\tilde{G}$. This is because for any element $h$ such that $ h~\mathcal {R}~id$ (i.e., $h$ has a right inverse), the map $m\mapsto m\cdot h$ from $M$ to $M$ is injective; thus, it is also surjective and so, $h$ has a left inverse too.} are right ideals. Then,  $\reachmf$ can be solved using a deterministic logspace algorithm when the accepting set $F$ is any of the $\mathcal{R}$-classes of $M$ other than its unit group $\tilde{G}$. 
\end{restatable}

Finally, we show the following classification of the complexity of the $\reachmf$ problem when the monoid $M$ is either $\BA_2$ or $\U$ (see preliminaries to see the connection between these monoids and the non-commutative aperiodic monoids) (Section~\ref{subsec: Dichotomy for BA_2 and U with Different Accepting Sets}). Note that by Theorem~\ref{thm: logspace algo for accepting identity for any finite monoid}, there is a logspace algorithm for the problem when $F=\{id\}$ or its complement by using the fact that the reachability problem reduces to its complement.

\begin{restatable}[]{theorem}{thmBAUsubset}
\label{thm:BA2Usubset}
     When the monoid $M$ is $\BA_2$ (or $\U$ respectively), $\reachmf$ is $\NL$-hard when accepting set $F$ contains at least one of $\alpha, \beta, \alpha\beta, \beta\alpha$ but not $0$; otherwise, $\reachmf$ is logspace solvable.
\end{restatable}

Table \ref{tab:summary_table} tabulates the above results and some results of \cite{RSS19}.


\begin{table}
\centering
\begin{tabularx}{\textwidth} { 
  | >{\centering\arraybackslash}X 
  | >{\centering\arraybackslash}X 
  | >{\centering\arraybackslash}X | }
 \hline
 \textbf{Restrictions on finite magma $M$} & \textbf{Restrictions on accepting set $F\subseteq M$} & \textbf{Space complexity of $\reachmf$ } \\
 \hline
 \hline
 Group, Quasigroup  & Any subset of $M$ & ~~Logspace\newline (Theorem 3.4, 3.7 in  \cite{RSS19}) \\
\hline
Commutative aperiodic monoid, Monoids in $\DA$  & Any subset of $M$ & ~~Logspace\newline (Theorem 5.1, 5.2 in \cite{RSS19}) \\
\hline
$\BA_2$  & $F=\{\alpha\beta\}$ & ~~~~~~$\NL$-hard\newline (\cite{KSS14} and \cite{RSS19}) \\
\hline
$\U$  & $F=\{\alpha\}$ & ~~$\NL$-hard\newline (Theorem 5.4 in \cite{RSS19}) \\
\hline
Monoid  & Any subset of unit group & ~~~~Logspace\newline(Theorem \ref{thm: logspace algo for accepting identity for any finite monoid})  \\
\hline
Commutative monoid  & Any subset of $M$ & ~~~~Logspace\newline (Theorem \ref{thm: logspace algo commutative monoids})\\
\hline
$\calL$ (or $\calR$)-commutative union-of-groups monoid  & Any subset of $M$  & ~~~~Logspace\newline (Theorem \ref{Logspace algo when R-classes = H-classes}) \\
\hline
Union-of-groups monoid where every $\calR$-class other than unit group is right ideal  & Any $\calR$-class of $M$ other than its unit group & ~~~~Logspace\newline (Theorem \ref{R-classes right ideal thm}) \\
\hline
Monoid with absorber $0$ and an idempotent $\mu\neq 0$ factorizable as $a\cdot b$ such that one of these holds true \newline (i) $b^2=0$ and $a^2=a$ \newline (ii) $a^2=b^2=0$ \newline (iii) $a^2=0$ and $b^2=b$ & $F=\{\mu\}$ &~~~~$\NL$-hard\newline(Theorem \ref{idempotent hardness with certain conditions})\\
\hline
 $\BA_2$ or $\U$  & $F=\{id\}$, or $F=M\setminus \{id\}$ & ~~~~Logspace\newline
 (Theorem \ref{thm: logspace algo for accepting identity for any finite monoid} and Remark \ref{remark:BA2-U}) \\
 \hline
$\BA_2$ or $\U$  & $F\subseteq M$ and $0 \in F$ & ~~~~Logspace\newline (Theorem \ref{thm:BA2Usubset}) \\
\hline
$\BA_2$ or $\U$  & $F\subseteq M$, $F\ne \{id\}$, and $0 \notin F$ & ~~~~$\NL$-hard\newline (Theorem \ref{thm:BA2Usubset}) \\
\hline
\end{tabularx}
\caption{Summary of space-complexity results discussed in this paper.}
\label{tab:summary_table}
\end{table}

\noindent{\bf Our Techniques:} Our main tool is the analysis of the product graph of a graph $G$ whose edges are labelled over the monoid $M$ (defined by \cite{RSS19}), which is a directed graph in which the paths correspond to walks in the original graph $G$. We use the structural properties of the monoid to derive the combinatorial structure of the graph and obtain upper and lower bounds.

\section{Preliminaries}
\label{ec:prelims}
This section introduces the terminology and notations required to understand the paper. We start by defining the space complexity classes. The class \ComplexityFont{L} is the set of all decision problems that can be solved by a deterministic Turing machine in logspace, and the class $\NL$ is the set of all decision problems that can be solved by a non-deterministic Turing machine in logspace. A problem $A$ is said to be $\NL$-hard if for every problem $B \in \NL$, there exists a logspace computable reduction function $\sigma:\{0,1\}^* \rightarrow \{0,1\}^*$, such that $\forall x,  ~ x \in B \iff \sigma(x) \in A$. A problem $A$ is said to be $\NL$-complete if $A \in \NL$, and it is $\NL$-hard. It is known that \ComplexityFont{L} $\subseteq \NL$, but the other containment is still unknown. Thus, we focus on the directed reachability problem (\textsc{Reach}) that captures all of $\NL$.  \cite{Rei04} gave a deterministic logspace algorithm for the undirected graph reachability problem, thus showing that the undirected reachability problem captures all of \ComplexityFont{L}. 
Reachability for a special class of directed graphs called Eulerian graphs is known to be in \ComplexityFont{L} \cite{RTV06}, where the Eulerian graph is a directed graph with the in-degree of each vertex in the graph equal to its out-degree. 


In a graph $G$, a \emph{walk} is a sequence of edges that joins a sequence of vertices. That is, a walk is of the form $(e_1, e_2, \ldots, e_r)$, where for each $1\leq i< r$, $e_i=(v_i, v_{i+1})\in E$; here, it is joining the sequence of vertices $(v_1, v_2, \ldots, v_{r+1})$. A \emph{path} is a walk in which vertices and edges are not allowed to repeat; that is, $e_1, e_2, \ldots, e_r$ are distinct edges, and $v_1, v_2, \ldots, v_{r+1}$ are distinct vertices.

Next, we define some algebraic structures relevant for the paper.
A \textit{semigroup} $S=(T,*)$ is a set $T$ with a binary operation $*$\footnote{Throughout the paper, we consider $(\cdot)$ as the binary operation} satisfying the properties of closure and associativity.
A \textit{monoid} $M=(T,*)$ is a semigroup with a binary operation $*$ where an identity exists in $T$. A \textit{group} $\mathcal{G} = (T, *)$ is a monoid with a binary operation $*$ where each element of $T$ has an inverse. A \textit{right absorber} $a$ is an element in a monoid $M$ such that for all $b \in M$,  $ba=a$. Similarly, a \textit{left absorber} $a \in M$ is such that for all $b \in M$, $ab=a$. An element that is both a right and a left absorber is simply called an \textit{absorber}. In a monoid, if an absorber exists, it is unique. An \textit{idempotent} element $\mu$ is an element of a monoid $M$ such that $\mu^2 = \mu$. An element $a$ of a monoid $M$ is said to be \textit{non-central} if there exists $x \in M$ such that $a \cdot x \neq x \cdot a$.
A subset $T$ of a monoid $M$ is said to be a \textit{right ideal} if $\forall a \in T$ and $m \in M$, the element $a \cdot m$ is in the set $T$. Similarly, a \textit{left ideal} can be defined. A subset $T$ of a monoid is said to be an \textit{ideal} if it is both a left and a right ideal. 

We study the reachability problem variant over different monoid classes. First, we look at some monoid classes defined via Green's relations.
There are five Green's relations $\calL,\calR,\calJ,\calH$ and $\mathcal{D}$ which are equivalence relations, see \cite{Col11}.
For a monoid $M$,
\begin{itemize}
    \item $a \leq_{\calL} b$ if $\exists x \in M$ such that $a = xb$.  
    
    $a \mathrel{\calL} b$ if $a \leq_{\calL} b$ and $b \leq_{\calL} a$\footnote{More generally, for any preorder $\leq$, it is a standard concept to define an equivalence relation (said to be \emph{induced} by $\leq$) wherein $a$ is related to $b$ when $a \leq b$ and $b \leq a$. Here, $\calL, \calR, \calJ$ are the equivalence relations induced by $\leq_{\calL}$, $\leq_{\calR}$, $\leq_{\calJ}$ respectively.}; In other words,
   $a \mathrel{\calL} b \iff M a = M b$
    
    \item $a \leq_{\calR} b$ if $\exists x \in M$ such that $a = bx$. 
    
    $a \mathrel{\calR} b$ if $a \leq_{\calR} b$ and $b \leq_{\calR} a$;  In other words, $a \mathrel{\calR} b \iff a M = b M$
    \item $a \leq_{\calJ} b$ if $\exists x,y \in M$ such that $a = xby$.
    
    $a \mathrel{\calJ} b$ if $a \leq_{\calJ} b$ and $b \leq_{\calJ} a$; In other words, $a \mathrel{\calJ} b \iff M a M = M b M$ 
    \item $\calH = \calL \cap \calR$. 
    
    $a \mathrel{\calH} b \iff a \mathrel{\calR} b$ and $a \mathrel{\calL} b$
    \item $\mathcal{D} = \calL \circ \calR = \calR \circ \calL$

    $a \mathrel{\calD} b \iff \exists x \in M \text{ such that } a \mathrel{\calL} x \text{ and } x \mathrel{\calR} b$
\end{itemize}

For finite monoids, the relations $\calD$ and $\calJ$ coincide. We say $a \le_{\calJ} b$ if $M a M \subseteq M b M$, and $a <_{\calJ} b$ if $M a M \subsetneq M b M$.

In a monoid $M$, a $\calH$-class $H$ contains an idempotent element if and only if $H$ forms a maximal subgroup of $M$. A monoid $M$ is a \textit{union-of-groups} (also called \textit{completely regular} monoid) if each $\calH$-class contains an idempotent element and thus forms a maximal subgroup of $M$. An example of a union-of-groups monoid is $M = \{a,a^2,e,ae\}$, where $a^2$ is the monoid identity and $\forall m \in M, em=e$, $a^2e=e$, $a^3 = a$. There are two $\calJ$  classes: one class with $\{a,a^2\}$, other with $\{e,ae\}$ where $e \mathrel{\calL} ae$, see  \cite{KTT07}. This monoid has three $\calH$-classes: $\{a,a^2\}$, $\{e\}$, $\{ae\}$ which are the maximal groups of $M$. A \textit{Clifford monoid} is a union-of-groups monoid whose idempotent elements commute.
Another monoid class we consider is the class of commutative monoids. A monoid is said to be  \textit{commutative} if $\forall a,b \in M$, we have $ab = ba$.  A monoid $M$ is said to be \textit{$\calL$-commutative} if for all $a,b \in M$, $ab$ is $\calL$ related to $ba$ (i.e. $ab$ and $ba$ belong to the same $\calL$-class), see \cite{nagy2001special}. We can similarly define \textit{$\calR$-commutative} monoids. See Figure~\ref {inclusion map figure} for inclusion of the monoid classes.

\FloatBarrier
\begin{figure}[H]
\label{inclusion map figure}
\centering
\begin{tikzpicture}[
    > = {Stealth[length=2mm, width=1.5mm]}, 
    every node/.style = {align=center, inner sep=4pt, font=\sffamily\small} 
  ]

  \node (top)      at (0, 3) {Union-of-groups \\ monoids};
  
  \node (Lcomm)    at (-2.8, 1.5) {$\mathcal{L}$-commutative \\ union-of-groups monoids};
  \node (Rcomm)    at (2.8, 1.5) {$\mathcal{R}$-commutative \\ union-of-groups monoids};
  
  \node (clifford) at (0, 0) {Clifford \\ monoids};
  
  \node (commU)    at (-2.8, -1.5) {Commutative \\ union-of-groups monoids};
  \node (groups)   at (2.8, -1.5) {Groups};
  
  \node (bot)      at (0, -3) {Commutative \\ groups};

  \draw[->] (bot)      -- (commU);
  \draw[->] (bot)      -- (groups);
  
  \draw[->] (commU)    -- (clifford);
  \draw[->] (groups)   -- (clifford);
  
  \draw[->] (clifford) -- (Lcomm);
  \draw[->] (clifford) -- (Rcomm);
  
  \draw[->] (Lcomm)    -- (top);
  \draw[->] (Rcomm)    -- (top);

\end{tikzpicture}
\caption{Overview of inclusions amongst certain sub-classes of union-of-groups monoids.}
\end{figure}

We now see some terminology related to $\calD$-classes. An element $x \in M$ is \emph{regular} if there exists $y \in M$ such that $xyx = x$. A $\calD$-class $A \subseteq M$ is regular if all of its elements are regular. Equivalently, a $\calD$-class is regular if it contains at least one idempotent element ($\mu^2 = \mu$).
A semigroup $S$ is called \textit{completely 0-simple} if (i) $S$ and $\{0\}$ are the only ideals (ii) $S^2 \ne \{0\}$ (iii) it contains a primitive idempotent (i.e. a non-zero idempotent $\mu$ that satisfies the following property for all non-zero idempotents $\gamma$: $\mu \gamma = \gamma \mu = \gamma \implies \mu = \gamma$).
A \textit{Principal factor} is defined as follows. Let $A$ be a regular $\calD$-class of a finite monoid $M$. 
Let $I_A = \{ x \in M \mid x <_{\calJ} a \text{ for any }a \in A \}$ be the ideal of all elements strictly below $A$ in the $\calJ$-ordering. The \emph{principal factor} associated with $A$ is the quotient semigroup $M(A) = (A \cup I_A) / I_A$, i.e., the Rees factor semigroup of the semigroup $A \cup I_A$ modulo its ideal $I_A$. In this quotient, all elements of $I_A$ are identified with a single absorbing element $0$; here, the product of any two elements $x,y$ is $x\cdot y$ (i.e., the same as the product in $M$) when $x, ~y,~ x\cdot y \in A$, and $0$ otherwise. For a regular $\calD$-class, $M(A)$ is a completely $0$-simple semigroup (see Lemma 2.4 in \cite{dandan2019group}).

Next we define \textit{Rees Matrix Semigroups \cite{How95}.}
Let $G$ be a finite group, and let $I$ and $\Lambda$ be finite index sets. Let $P = (P_{\lambda, i})$ be a $\Lambda \times I$ matrix with entries in $G \cup \{0\}$, called the \emph{sandwich matrix}. The \emph{Rees matrix semigroup} $\mathcal{M}^0(G; I, \Lambda; P)$ is the set $(I \times G \times \Lambda) \cup \{0\}$ equipped with the binary operation defined by:
\[
(i, g, \lambda)(j, h, \kappa) = 
\begin{cases} 
(i, g P_{\lambda, j} h, \kappa) & \text{if } P_{\lambda, j} \neq 0 \\
0 & \text{if } P_{\lambda, j} = 0
\end{cases}
\]
and $x \cdot 0 = 0 \cdot x = 0$ for all $x$. 

By the Rees Theorem, every completely $0$-simple semigroup is isomorphic to a Rees matrix semigroup where every row and column of $P$ contains at least one non-zero entry. 
The class $\DS$ consists of all finite monoids $M$ in which every regular $\calD$-class forms a subsemigroup of $M$ \cite{Pin86}. A regular $\calD$-class $A$ is a subsemigroup if and only if its corresponding sandwich matrix $P$ contains no zero entries. Thus, $M \notin \DS$ if and only if there exists a regular $\calD$-class $A$ whose sandwich matrix $P$ contains at least one zero entry.

Next, we define another monoid class over which we study the reachability problem.  A monoid is called \textit{aperiodic} if for every element $a \in M$, there exists a $k \in \mathbb{N}$ such that $a^k=a^{k+1}$. The number $k$ is called the \textit{index} of the element $a$ in $M$. A monoid $N$ is said to \textit{divide} a monoid $M$ if there is a surjective homomorphism (say $f$) from a submonoid (say $M'$) of $M$ to $N$. $\DA$ is a class of aperiodic monoids such that if there exists an idempotent element in a $\calJ$-class $J$, then all the elements of $J$ are idempotent elements. 

Further, among aperiodic monoids, we primarily focus on two non-commutative monoids, which are $\BA_2$ \footnote{More generally, for every $k\geq 1$, the \emph{aperiodic Brandt monoid} of size $k$ (denoted as $\BA_k$) is the monoid that consists of all partial functions $f:[k]\rightarrow [k]$ such that $|f^{-1}([k])|\leq 1$ along with the identity map, with function composition as the monoid operation, see \cite{MP84}}
and $\U$ \footnote{some literature also uses the notation $B_2^1$ and $A_2^1$ for $\BA_2$ and $\U$ respectively} 
(see table~\ref{tab:BA_2U} for the Cayley tables). It is known (see \cite{Tes98thesis}) that for any nontrivial finite non-commutative aperiodic monoid $M$ which is outside the pseudovariety $\DA$, either $\BA_2$ or $\U$ divides $M$. In terms of formal language theory, these two monoids are the syntactic monoids (see \cite{Pin86}) of the regular languages represented by $(c^*ac^*bc^*)^*$ and $((b+c)^*a(b+c)^*b(b+c)^*)^*$, respectively.

For the purpose of this paper, we require only the Cayley tables for the monoids $\BA_2$ and $\U$, which are given below.
\begin{table}
\vspace{-8mm}
\begin{center}
\begin{minipage}
{0.45\textwidth}

\centering
\[
\begin{array}{c|cccccc}
   & 0 & 1 & \alpha & \beta & \alpha\beta & \beta\alpha \\ \hline
0  & 0 & 0 & 0 & 0 & 0 & 0 \\
1  & 0 & 1 & \alpha & \beta & \alpha\beta & \beta\alpha \\
\alpha & 0 & \alpha & 0 & \alpha\beta & 0 & \alpha \\
\beta  & 0 & \beta & \beta\alpha & 0 & \beta & 0 \\
\alpha\beta & 0 & \alpha\beta & \alpha & 0 & \alpha\beta & 0 \\
\beta\alpha & 0 & \beta\alpha & 0 & \beta & 0 & \beta\alpha \\
\end{array}
\]
\end{minipage}
\begin{minipage}{0.45\textwidth}
\centering
\hspace{-30 mm}
\hspace{-30 mm}
\[
\begin{array}{c|cccccc}
   & 0 & 1 & \alpha & \beta & \alpha\beta & \beta\alpha \\ \hline
0  & 0 & 0 & 0 & 0 & 0 & 0 \\
1  & 0 & 1 & \alpha & \beta & \alpha\beta & \beta\alpha \\
\alpha & 0 & \alpha & \alpha & \alpha\beta & \alpha\beta & \alpha \\
\beta  & 0 & \beta & \beta\alpha & 0 & \beta & 0 \\
\alpha\beta & 0 & \alpha\beta & \alpha & 0 & \alpha\beta & 0 \\
\beta\alpha & 0 & \beta\alpha & \beta\alpha & \beta & \beta & \beta\alpha \\
\end{array}
\]
\end{minipage}
\end{center}
\caption{Cayley tables for $\BA_2$ and $\U$ (respectively).}
\label{tab:BA_2U}
\end{table}


The \textit{unit group} of a monoid $M$ is the set of all units/invertible elements of the monoid, i.e., all elements $a\in M$ for which there exists $x\in M$ such that $a\cdot x= x\cdot a = id$, where $id$ denotes the monoid identity of $M$. For any graph $A$, its \textit{endomorphism monoid} is the set of all homomorphisms from $A$ to itself, with function composition as the monoid operation.  \\[-2mm]

\noindent{\bf The Product Graph:}
We now define our main tool called the product graph (defined in \cite{RSS19}). Let $M$ be a fixed finite monoid. Consider a graph $G(V, E)$ with labelling $\phi: E \rightarrow M$. The product graph $G'(V',E')$ of $G$ is the directed graph with $V' = V \times M$ and $E' = \{ ((u, m),(v,m')) \mid m,m'\in M, \{u,v\} \in E \textrm{ and } m \cdot \phi(\{u,v\}) = m'\}$.
\begin{proposition}[\cite{RSS19}] 
    Let $s,t \in V, m \in M$ and $id$ is the identity element of $M$. Then there is a walk from $s$ to $t$ in $G$ with yield $m$ if and only if there is a path from $(s, id)$ to $(t,m)$ in $G'$. 
\end{proposition}

\section{Reachability with Monoid Identity/Idempotent Element as the Accepting Element}

In this section, we prove the results that we described in the introduction for the case when the accepting subset is the \textit{identity} of the monoid (Section~\ref{subsec: Logspace Algorithm for all Finite Monoids}) and for the accepting subset being an \textit{idempotent} element of the monoid (Section~\ref{subsec: NL-hardness with an Idempotent as the Accepting Element}).

\subsection{Logspace Algorithm with Monoid Identity as the Accepting Element}
\label{subsec: Logspace Algorithm for all Finite Monoids}
In the following theorem, we show a logspace algorithm for any fixed finite monoid when the accepting set is the identity element of the monoid.

\monoididentityalgo*

\begin{proof}
Drop all edges with labels in $M\setminus \tilde{G}$ and run the logspace algorithm by \cite{RSS19} for undirected labelled reachability over group $\tilde{G}$ on the graph so obtained. We now argue that dropping such edges is safe, i.e., any walk in the original graph that is lost now (due to some of its edges being dropped) had its yield outside $F$. Consider any walk $P$ that contains an edge labelled with an element from $M\setminus \tilde{G}$; let $e$ denote the first such edge of $P$, say with label $a\in M\setminus \tilde{G}$.  It suffices
to argue that $P$’s yield can never be in $\tilde{G}$. For the sake of contradiction, assume that $P$’s yield is $g$ for some $g\in \tilde{G}$. Let $x \in \tilde{G}$ and $y \in M$ denote the yields of subwalks of $P$ that appear before edge
$e$ and after edge $e$, respectively. Then, we have $x\cdot a \cdot y = g$. Pre-multiplying and post-multiplying both sides by $x^{-1}$ and $g^{-1}\cdot x$ respectively, we get $a\cdot (y\cdot g^{-1}\cdot x) = id$. Post-multiplying both sides by $a$, we get $a\cdot (y\cdot g^{-1}\cdot x\cdot a) = a$. Note that we also have $a\cdot (id) = a$. Now, if $y\cdot g^{-1}\cdot x\cdot a = id$, then $a$ has $y\cdot g^{-1}\cdot x$ as both its left inverse and right inverse, which is not possible as $a$ is a non-unit. So, we get $y\cdot g^{-1}\cdot x\cdot a \neq id$.   Thus, the sequence $a \cdot m_1, a \cdot m_2, ...., a \cdot m_k$ (where $m_1, m_2, ...., m_k$ denote all elements of $M$) has the element $a$ repeated at least twice. So, this sequence must skip at least one other element (say $m'$) of $M$.\footnote {One may also view this as saying that while embedding the monoid $M$ into the full transformation monoid of $M$ as per Cayley's theorem, the map from $M$ to $M$ that is associated with the element $a$ is a non-injective (and so, also non-surjective) map.} However, $m'$ can be written as $a\cdot (y\cdot g^{-1}\cdot x\cdot m')$ and so, $m'$ should not have been skipped by this sequence, a contradiction. This proves Theorem~\ref{thm: logspace algo for accepting identity for any finite monoid}. 
\end{proof} 

We remark that the proof of Theorem~\ref{thm: logspace algo for accepting identity for any finite monoid} can also be viewed as arguing that $M\setminus \tilde{G}$ forms an ideal.

We now give logspace algorithms for special classes of monoids when the accepting set is the identity element. The arguments use the underlying structure of the monoid.
\subsection*{Union-of-Groups monoids}

\begin{theorem}
\label{thm: accepting set is in unit group for union-of-groups monoid}
Consider any finite fixed union-of-groups monoid $M$. Let $\tilde{G}$ denote the unit group (i.e., the $\mathcal{H}$-class containing the monoid identity $id$) of $M$. Then, $\textsc{Reach}_{M, F}$ can be solved using a deterministic logspace algorithm for all accepting sets $F\subseteq \tilde{G}$.
\end{theorem}

\begin{proof}
We drop all edges with labels in $M\setminus \tilde{G}$ and simply run the logspace algorithm by \cite{RSS19} for undirected reachability over group $\tilde{G}$ on the graph so obtained. We now argue that dropping such edges is safe, i.e., any walk in the original graph that is lost now (due to some of its edges being dropped) had its yield outside $F$. Consider any walk $P$ that contains an edge (say $e$) labelled with an element (say $a$) from $M\setminus \tilde{G}$. It suffices to argue that $P$'s yield can never be in $\tilde{G}$. For the sake of contradiction, assume that $P$'s yield is $g$ for some $g\in \tilde{G}$. Let $x\in M$ and $y\in M$ denote the yields of subwalks of $P$ that appear before edge $e$ and after edge $e$, respectively. Then,  we have $x\cdot a \cdot y = g$. Consider the following two cases:
\begin{description}
    \item[\textbf{Case 1}: $x\cdot a\in \tilde{G}$.]
    Then, pre-multiplying both sides of $x\cdot a\cdot y=g$ by $(x\cdot a)^{-1}$ gives $y = (x\cdot a)^{-1}\cdot g$ and so, $y\in \tilde{G}$. If $x\in \tilde{G}$, then pre-multiplying and post-multiplying $x\cdot a\cdot y = g$ by $x^{-1}$ and $y^{-1}$ respectively would give $a=x^{-1}\cdot g\cdot y^{-1}\in \tilde{G}$, which is not possible as $a\not\in \tilde{G}$. So, we have $x\notin \tilde{G}$. Let $H$ denote the group containing $x$, and let $id_H$ denote the identity of $H$. Next, pre-multiplying and post-multiplying both sides of $x\cdot a\cdot y = g$ by $id_H$ and $g^{-1}$ respectively gives $x\cdot a\cdot y\cdot g^{-1} = id_H$. However, LHS = $(x\cdot a\cdot y)\cdot g^{-1} = g\cdot g^{-1} = id\neq id_H$ and so, LHS $\neq$ RHS, a contradiction.
    \item[\textbf{Case 2}: $x\cdot a \not\in \tilde{G}$.]
    Let $H$ denote the group containing $x\cdot a$, and let $id_H$ denote the identity of $H$. Pre-multiplying and post-multiplying both sides of $x\cdot a\cdot y=g$ by $e_H$ and $g^{-1}$ respectively gives $x\cdot a\cdot y\cdot g^{-1} = id_H$. However, LHS = $(x\cdot a\cdot y)\cdot g^{-1} = g\cdot g^{-1} = id\neq id_H$ and so, LHS $\neq$ RHS, a contradiction.
\end{description}
This proves Theorem~\ref{thm: accepting set is in unit group for union-of-groups monoid}.
\end{proof}

\subsection*{Aperiodic monoids}

\begin{theorem}
\label{thm: aperiodic monoid logspace algo for identity}
Consider any finite fixed aperiodic monoid $M$ with monoid identity $id$. Then, $\reachmf$ with $F=\{id\}$ can be solved using a deterministic logspace algorithm.
\end{theorem}

\begin{proof}
We drop all edges with labels $\neq id$ and run the logspace algorithm for undirected, unlabelled graphs on the resulting graph. We now argue that dropping such edges is safe, i.e., any walk in the original graph that is lost now (due to some of its edges being dropped) had its yield outside $F$. Consider any walk $P$ that contains at least one edge labelled with an element from $M\setminus\{id\}$; let $e$ denote the first such edge, and let $a\in M\setminus \{id\}$ denote its label. It suffices to argue that $P$'s yield can never be $id$. For the sake of contradiction, assume that $P$'s yield is $id$. The yield of the subwalk of $P$ before edge $e$ is $id$. Let $x$ denote the yield of a subwalk of $P$ after edge $e$. Then, we have $a\cdot x = id$. As $M$ is aperiodic, there exists $\ell\geq 1$ such that $a^{\ell}=a^{\ell+1}$ and $x^{\ell}=x^{\ell+1}$. Note that $a^{\ell}$ and $x^{\ell}$ are idempotent elements. Let $\tilde{a}:=a^{\ell}$ and $\tilde{x}:=x^{\ell}$. Because of associativity, we write $\tilde{a}\cdot \tilde{x} = a^{\ell}\cdot x^{\ell} = a\cdot (a\cdot (\ldots a\cdot (a\cdot x )\cdot x\ldots )\cdot x)\cdot x$, which is $id$. Now, again because of associativity, we  have $\tilde{a}\cdot (\tilde{a}\cdot \tilde{x}) = (\tilde{a}\cdot\tilde{a})\cdot \tilde{x}$. Note that LHS $= \tilde{a}\cdot id$ (as $\tilde{a}\cdot \tilde{x}=id$), which is $\tilde{a}$; also, RHS $= \tilde{a}\cdot \tilde{x}$ (because $\tilde{a}$ is idempotent element), which in turn is $id$. Therefore, we get $\tilde{a}= id$. That is, $a^{\ell} = id$. However, as we also have $a^{\ell+1}=a^{\ell}$, it follows that $a=id$, a contradiction. This proves Theorem~\ref{thm: aperiodic monoid logspace algo for identity}.
\end{proof}

\subsection*{Commutative monoids}

\begin{theorem}
\label{thm: commutative monoids accepting set in unit group algo}
Consider any fixed finite commutative monoid $M$. Let $\tilde{G}$ denote the unit group (i.e., the set of all invertible elements) of $M$. Then,
$\reachmf$ can be solved using a deterministic logspace algorithm for all accepting sets $F\subseteq \tilde{G}$.
\end{theorem}

\begin{proof}
We drop all edges with labels in $M\setminus \tilde{G}$ and simply run the logspace algorithm by \cite{RSS19} for undirected labelled reachability over group $\tilde{G}$ on the graph so obtained. We now argue that dropping such edges is safe, i.e., any walk in the original graph that is lost now (due to some of its edges being dropped) had its yield outside $F$. Consider any walk $P$ that contains at least one edge labelled with an element from $M\setminus \tilde{G}$; let $e$ denote the first such edge, and let $a\in M\setminus \tilde{G}$ denote its label. It suffices to argue that $P$'s yield can never be in $\tilde{G}$. For the sake of contradiction, assume that $P$'s yield is $g$ for some $g\in \tilde{G}$. Let $x\in \tilde{G}$ and $y\in M$ denote the yields of subwalks of $P$ that appear before edge $e$ and after edge $e$, respectively. Then,  $x\cdot a \cdot y = g$. Pre-multiplying both sides by $x^{-1}$ and post-multiplying both sides by $g^{-1}\cdot x$ gives $a\cdot y\cdot g^{-1}\cdot x = id$. So, $y\cdot g^{-1}\cdot x$ is the right inverse (and so, also the left inverse because $M$ is commutative) of $a$. That is, $a$ is a unit of $M$, which is not possible because $a\not\in \tilde{G}$. This proves Theorem~\ref{thm: commutative monoids accepting set in unit group algo}.
\end{proof}

\subsection{NL-hardness with an Idempotent Element as the Accepting Element}
\label{subsec: NL-hardness with an Idempotent as the Accepting Element}
Complementing the above, we show that there are non-commutative monoids (e.g., $\BA_2$ and $\U$) where the problem is $\NL$-hard even when the accepting element is an idempotent element. We show the following theorem from the introduction.

\specialidempotent*

\begin{proof}
This proof builds on the proof of Theorem 5.4 in \footnote{the conference version of this paper (i.e., \cite{KMS26}) claimed and proved this theorem (without the conditions that $a$ and $b$ are idempotents in Cases (i) and (iii) respectively). However, an anonymous reviewer later pointed out an issue with this proof. Herein, we have fixed this issue by adding the conditions that $a$ and $b$ are idempotents in Cases (i) and (iii) respectively, modifying the reduction, and also rectifying the analysis of the reverse direction.} in \cite{RSS19}. The proof will demonstrate $\NL$-hardness in each of the above cases separately.

\paragraph{\underline{Case (i): $b^2=0$ and $a$ is idempotent:}} We reduce from the Directed Reachability problem. Given an instance $(G, s, t)$ of Directed Reachability, drop all incoming edges at $s$ and drop all outgoing edges at $t$; this gives an equivalent instance, as no directed $s$-$t$ path in $G$ can make use of such dropped edges. Then, construct a reduced graph $G'$ by splitting every edge of $G$ and assigning labels $a$, $ba$ and $b$ to the three edges so obtained. That is, for every edge $(u,v)\in E(G)$, add (apart from $u$ and $v$) in $G'$ two new vertices (say $m_{uv}^{(1)}$ and $m_{uv}^{(2)}$ respectively) and edges $\{u,m_{uv}^{(1)}\}$, $\{m_{uv}^{(1)}, m_{uv}^{(2)}\}$ and $\{m_{uv}^{(2)}, v\}$; also, set their labels $\phi(\{u,m_{uv}^{(1)}\}) = a$, $\phi(\{m_{uv}^{(1)}, m_{uv}^{(2)}\}) = ba$ and $\phi(\{m_{uv}^{(2)}, v\}) = b$. We argue that $G$ has an $s$-to-$t$ directed path if and only if $G'$ has an $s$-to-$t$ walk of yield $\mu$. 

If $G$ has an $s$-to-$t$ path $P$, the corresponding $s$-to-$t$ path (let length($P$) = $\ell$) in $G'$ has yield $(a\cdot b)^{2\ell} = \mu^{2\ell}$, which is $\mu$ because $\mu$ is an idempotent element. 

Next, suppose that $G'$ has an $s$-to-$t$ walk of yield $\mu$; let $P$ be a walk of minimum length amongst all such walks. For contradiction, assume that $G$ does not have a directed $s$-to-$t$ path. Then, the walk $P$ in $G'$ does not translate into a directed $s$-$t$ walk in $G$. So, starting from $s$, the walk $P$ must have picked a first `wrong' edge. That is, before this wrong edge was picked, the subwalk of $P$ in $G'$ so far translated back to a directed walk in $G$; once the wrong edge was picked, the corresponding walk in $G$ could not be extended in a directed manner. Also, as $s$ has no incoming edges, the first edge of $P$ cannot be a wrong edge. Now, the first wrong edge of $P$ in $G'$ is of the form $\{v,m_{uv}^{(2)}\}$ (Case 1) or $\{m_{uv}^{(1)},u\}$ (Case 2) or $\{m_{uv}^{(2)}, m_{uv}^{(1)}\}$ (Case 3)  for some edge $(u, v) \in E(G)$.
\begin{description}
    \item[Case 1: The first wrong edge of $P$ is of the form $\{v,m_{uv}^{(2)}\}$ for some edge $(u,v)\in E(G)$]:
    
    As $\{v,m_{uv}^{(2)}\}$ is the first wrong edge picked, the edge just before it (in the walk $P$ in $G'$), being a correct edge, must be of the form $\{m_{wv}^{(2)}, v\}$ for some edge $(w, v)\in E(G)$ (Refer Fig. \ref{fig:Case (i) Case 1} and Fig. \ref{fig:path_traversal_legend} for legend). Now, both the first wrong edge and the correct edge (picked by $P$ just before the first wrong edge) have label $b$. That is, $\phi(\{v,m_{uv}^{(2)}\}) = \phi(\{m_{wv}^{(2)}, v\}) = b$. So, $P$'s yield must have $b^2$ as a factor. As $b^2=0$ is the absorber, this forces $P$'s yield to $0$, which is not possible as $P$'s yield is $\mu\neq 0$.

\begin{figure}[htbp]
\centering
\begin{tikzpicture}
    \draw[black, thick] (0, 0) -- (0.8, 0);
    \node[right, font=\small] at (0.9, 0) {Edges of $G$};

    \draw[red, thick, postaction={decorate, decoration={markings, mark=at position 0.6 with {\arrow[scale=1.5, red]{>}}}}] (4.3, 0) -- (5.1, 0);
    \node[right, font=\small] at (5.2, 0) {Path Edge (Correct)};

    \draw[red, thick, dashed, postaction={decorate, decoration={markings, mark=at position 0.6 with {\arrow[scale=1.5, red]{>}}}}] (9.1, 0) -- (9.9, 0);
    \node[right, font=\small] at (10.0, 0) {Path Edge (Wrong)};
\end{tikzpicture}
\caption{Legend for Figure~\ref{fig:Case (i) Case 1} to Figure~\ref{fig:Case (i) Case 3}.}
\label{fig:path_traversal_legend}
\end{figure}

    \begin{figure}[htbp]
\centering
\begin{tikzpicture}[
    scale=1.5,
    black dot/.style={circle, fill=black, inner sep=1.2pt},
    black dot big/.style={circle, fill=black, inner sep=2.2pt},
    red dot/.style={circle, fill=red, inner sep=1.2pt},
    red dot big/.style={circle, fill=red, inner sep=2.2pt},
    black edge/.style={black, thick},
    path arrow/.style={
        postaction={decorate, decoration={markings, mark=at position 0.55 with {\arrow[scale=1.5, red]{>}}}}
    }
]

    \coordinate (U)      at (0, 0);
    \coordinate (M1_uv) at (1.5, 0);
    \coordinate (M2_uv) at (3.0, 0);
    \coordinate (V)      at (4.5, 0);
    \coordinate (M2_wv) at (6.0, 0);
    \coordinate (M1_wv) at (7.5, 0);
    \coordinate (W)      at (9.0, 0);

    \draw[black edge] (U) -- (M1_uv) node[midway, below=4pt] {$a$};
    \draw[black edge] (M1_uv) -- (M2_uv) node[midway, below=4pt] {$ba$};
    \draw[black edge] (M2_uv) -- (V) node[midway, below=4pt] {$b$};
    \draw[black edge] (V) -- (M2_wv) node[midway, below=4pt] {$b$};
    \draw[black edge] (M2_wv) -- (M1_wv) node[midway, below=4pt] {$ba$};
    \draw[black edge] (M1_wv) -- (W) node[midway, below=4pt] {$a$};

    \node[black dot] at (U) {};
    \node[above=6pt of U] {$u$};

    \node[black dot] at (M1_uv) {};
    \node[above=6pt of M1_uv] {$m_{uv}^{(1)}$};

    \node[black dot] at (M2_uv) {};
    \node[above=6pt of M2_uv] {$m_{uv}^{(2)}$};

    \node[black dot big] at (V) {};
    \node[above=6pt of V] {$v$};

    \node[black dot] at (M2_wv) {};
    \node[above=6pt of M2_wv] {$m_{wv}^{(2)}$};

    \node[black dot] at (M1_wv) {};
    \node[above=6pt of M1_wv] {$m_{wv}^{(1)}$};

    \node[black dot] at (W) {};
    \node[above=6pt of W] {$w$};

    \coordinate (L1_M2_wv) at (6.0, -1.2);
    \coordinate (L1_V)     at (4.5, -1.2);
    \coordinate (L1_M2_uv) at (3.0, -1.2);

    \draw[red, thick, path arrow] (L1_M2_wv) -- (L1_V) node[midway, below=2pt] {$b$};

    \draw[red, thick, dashed, path arrow] (L1_V) -- (L1_M2_uv) node[midway, below=2pt] {$b$};

    \node[red dot] at (L1_M2_wv) {};
    \node[above=3pt of L1_M2_wv, red] {$m_{wv}^{(2)}$};

    \node[red dot big] at (L1_V) {};
    \node[above=3pt of L1_V, red] {$v$};

    \node[red dot] at (L1_M2_uv) {};
    \node[above=3pt of L1_M2_uv, red] {$m_{uv}^{(2)}$};

\end{tikzpicture}
\caption{Figure for Case 1 of Case (i) in Theorem~\ref{idempotent hardness with certain conditions}.}
\label{fig:Case (i) Case 1}
\end{figure}

    \item[Case 2: The first wrong edge of $P$ is of the form $\{m_{uv}^{(1)},u\}$ for some edge $(u,v)\in E(G)$]:\\ As $\{m_{uv}^{(1)}, u\}$ is the first wrong edge picked, the edge just before it (in the walk $P$ in $G'$), being a right edge, must be $\{u, m_{uv}^{(1)}\}$. Now, we make two subcases based on whether the edge just after the first wrong edge is a right edge or a wrong edge. 
    \begin{description}
        \item[Case 2a:  At $u$, just after the first wrong edge, the next edge taken by $P$ is a right edge]:\\
        Then, this right edge (just after the first wrong edge) must be of the form $\{u,m_{uw}^{(1)}\}$ (which is labelled $a$) for some edge $(u,w) \in E(G)$ (Refer Fig. \ref{fig:Case (i) Case 2a}). Now, construct a new $s$-$t$ walk $\tilde{P}$ from $P$ obtained by bypassing two of its consecutive edges, namely the right edge just before the first wrong edge (i.e., the edge $\{u,m_{uv}^{(1)}\}$ labelled $a$), and the first wrong edge (i.e., the edge $\{m_{uv}^{(1)},u\}$ labelled $a$). Note that the edge $\{u,m_{uw}^{(1)}\}$ just after these two bypassed edges in $P$ has label $a$; so, after bypassing, in the expression of walk's yield, the factor $a\cdot a\cdot a$ (contributed by the edges $\{u,m_{uv}^{(1)}\}, \{m_{uv}^{(1)},u\}, \{u,m_{uw}^{(1)}\}$ in $P$) has been replaced by just a single $a$ (contributed by the edge $\{u,m_{uw}^{(1)}\}$ in $\tilde{P}$). So, as $a^3=a$ (since $a$ is idempotent), yield of $\tilde{P}$ is same as that of $P$ (i.e., $\mu$). However, $\tilde{P}$ also has strictly smaller length than $P$; this is not possible as $P$ has minimum length amongst all $s$-$t$ walks in $G'$ of yield $\mu$.

        \begin{figure}[htbp]
\centering
\begin{tikzpicture}[
    scale=1.5,
    black dot/.style={circle, fill=black, inner sep=1.2pt},
    black dot big/.style={circle, fill=black, inner sep=2.2pt},
    red dot/.style={circle, fill=red, inner sep=1.2pt},
    red dot big/.style={circle, fill=red, inner sep=2.2pt},
    black edge/.style={black, thick},
    path arrow/.style={
        postaction={decorate, decoration={markings, mark=at position 0.55 with {\arrow[scale=1.5, red]{>}}}}
    }
]

    \coordinate (W)      at (0, 0);
    \coordinate (M2_uw) at (1.5, 0);
    \coordinate (M1_uw) at (3.0, 0);
    \coordinate (U)     at (4.5, 0);
    \coordinate (M1_uv) at (6.0, 0);
    \coordinate (M2_uv) at (7.5, 0);
    \coordinate (V)      at (9.0, 0);

    \draw[black edge] (W) -- (M2_uw) node[midway, below=4pt] {$b$};
    \draw[black edge] (M2_uw) -- (M1_uw) node[midway, below=4pt] {$ba$};
    \draw[black edge] (M1_uw) -- (U) node[midway, below=4pt] {$a$};
    \draw[black edge] (U) -- (M1_uv) node[midway, below=4pt] {$a$};
    \draw[black edge] (M1_uv) -- (M2_uv) node[midway, below=4pt] {$ba$};
    \draw[black edge] (M2_uv) -- (V) node[midway, below=4pt] {$b$};

    \node[black dot] at (W) {};
    \node[above=6pt of W] {$w$};

    \node[black dot] at (M2_uw) {};
    \node[above=6pt of M2_uw] {$m_{uw}^{(2)}$};

    \node[black dot] at (M1_uw) {};
    \node[above=6pt of M1_uw] {$m_{uw}^{(1)}$};

    \node[black dot big] at (U) {};
    \node[above=6pt of U] {$u$};

    \node[black dot] at (M1_uv) {};
    \node[above=6pt of M1_uv] {$m_{uv}^{(1)}$};

    \node[black dot] at (M2_uv) {};
    \node[above=6pt of M2_uv] {$m_{uv}^{(2)}$};

    \node[black dot] at (V) {};
    \node[above=6pt of V] {$v$};

    \coordinate (L1_U)     at (4.5, -0.9);
    \coordinate (L1_M1_uv) at (6.0, -0.9);

    \coordinate (L2_M1_uv) at (6.0, -1.6);
    \coordinate (L2_U)     at (4.5, -1.6);
    \coordinate (L2_M1_uw) at (3.0, -1.6);

    \draw[red, thick, path arrow] (L1_U) -- (L1_M1_uv) node[midway, below=2pt] {$a$};

    \node[red dot big] at (L1_U) {};
    \node[above=3pt of L1_U, red] {$u$};

    \node[red dot] at (L1_M1_uv) {};
    \node[above=3pt of L1_M1_uv, red] {$m_{uv}^{(1)}$};

    \draw[red, dashed] (L1_M1_uv) -- (L2_M1_uv);

    \draw[red, thick, dashed, path arrow] (L2_M1_uv) -- (L2_U) node[midway, below=2pt] {$a$};

    \draw[red, thick, path arrow] (L2_U) -- (L2_M1_uw) node[midway, below=2pt] {$a$};

    \node[red dot] at (L2_M1_uv) {};
    \node[right=3pt of L2_M1_uv, red] {$m_{uv}^{(1)}$};

    \node[red dot big] at (L2_U) {};
    \node[above=3pt of L2_U, red] {$u$};

    \node[red dot] at (L2_M1_uw) {};
    \node[above=3pt of L2_M1_uw, red] {$m_{uw}^{(1)}$};

\end{tikzpicture}
\caption{Figure for Case 2a of Case (i) in Theorem~\ref{idempotent hardness with certain conditions}.}
\label{fig:Case (i) Case 2a}
\end{figure}

\color{black}
\item[Case 2b: At $u$, just after the first wrong edge, the next edge taken by $P$ is a wrong edge]:\\
Then, this second wrong edge of $P$ must be of the form $\{u,m_{wu}^{(2)}\}$ (which is labelled $b$) for some edge $(w,u) \in E(G)$. Again, we make two sub-subcases based on which of the two edges incident to $m_{wu}^{(2)}$ the path $P$ takes just after this second wrong edge.
\begin{description}
\item[Case 2ba: At $m_{wu}^{(2)}$, just after the second wrong edge, $P$ takes the edge $\{m_{wu}^{(2)},u\}$]:\\
Then, the consecutive edges $\{u,m_{wu}^{(2)}\}$ and $\{m_{wu}^{(2)},u\}$ in $P$ (both are labelled $b$) contribute a factor of $b^2=0$ in the expression of walk $P$'s yield, thus making the yield of the walk $P$ to be $0$; this is not possible as the yield of $P$ is $\mu\neq 0$ instead (Refer Fig. \ref{fig:Case (i) Case 2ba}).

\begin{figure}[htbp]
\centering
\begin{tikzpicture}[
    scale=1.5,
    black dot/.style={circle, fill=black, inner sep=1.2pt},
    black dot big/.style={circle, fill=black, inner sep=2.2pt},
    red dot/.style={circle, fill=red, inner sep=1.2pt},
    red dot big/.style={circle, fill=red, inner sep=2.2pt},
    black edge/.style={black, thick},
    path arrow/.style={
        postaction={decorate, decoration={markings, mark=at position 0.55 with {\arrow[scale=1.5, red]{>}}}}
    }
]

    \coordinate (W)      at (0, 0);
    \coordinate (M1_wu) at (1.5, 0);
    \coordinate (M2_wu) at (3.0, 0);
    \coordinate (U)     at (4.5, 0);
    \coordinate (M1_uv) at (6.0, 0);
    \coordinate (M2_uv) at (7.5, 0);
    \coordinate (V)      at (9.0, 0);

    \draw[black edge] (W) -- (M1_wu) node[midway, below=4pt] {$a$};
    \draw[black edge] (M1_wu) -- (M2_wu) node[midway, below=4pt] {$ba$};
    \draw[black edge] (M2_wu) -- (U) node[midway, below=4pt] {$b$};
    \draw[black edge] (U) -- (M1_uv) node[midway, below=4pt] {$a$};
    \draw[black edge] (M1_uv) -- (M2_uv) node[midway, below=4pt] {$ba$};
    \draw[black edge] (M2_uv) -- (V) node[midway, below=4pt] {$b$};

    \node[black dot] at (W) {};
    \node[above=6pt of W] {$w$};

    \node[black dot] at (M1_wu) {};
    \node[above=6pt of M1_wu] {$m_{wu}^{(1)}$};

    \node[black dot] at (M2_wu) {};
    \node[above=6pt of M2_wu] {$m_{wu}^{(2)}$};

    \node[black dot big] at (U) {};
    \node[above=6pt of U] {$u$};

    \node[black dot] at (M1_uv) {};
    \node[above=6pt of M1_uv] {$m_{uv}^{(1)}$};

    \node[black dot] at (M2_uv) {};
    \node[above=6pt of M2_uv] {$m_{uv}^{(2)}$};

    \node[black dot] at (V) {};
    \node[above=6pt of V] {$v$};

    \coordinate (L1_U)     at (4.5, -0.9);
    \coordinate (L1_M1_uv) at (6.0, -0.9);

    \coordinate (L2_M1_uv) at (6.0, -1.6);
    \coordinate (L2_U)     at (4.5, -1.6);
    \coordinate (L2_M2_wu) at (3.0, -1.6);

    \coordinate (L3_M2_wu) at (3.0, -2.4);
    \coordinate (L3_U)     at (4.5, -2.4);

    \draw[red, thick, path arrow] (L1_U) -- (L1_M1_uv) node[midway, below=2pt] {$a$};

    \node[red dot big] at (L1_U) {};
    \node[above=3pt of L1_U, red] {$u$};

    \node[red dot] at (L1_M1_uv) {};
    \node[above=3pt of L1_M1_uv, red] {$m_{uv}^{(1)}$};

    \draw[red, dashed] (L1_M1_uv) -- (L2_M1_uv);

    \draw[red, thick, dashed, path arrow] (L2_M1_uv) -- (L2_U) node[midway, below=2pt] {$a$};
    \draw[red, thick, dashed, path arrow] (L2_U) -- (L2_M2_wu) node[midway, below=2pt] {$b$};

    \node[red dot] at (L2_M1_uv) {};
    \node[right=3pt of L2_M1_uv, red] {$m_{uv}^{(1)}$};

    \node[red dot big] at (L2_U) {};
    \node[above=3pt of L2_U, red] {$u$};

    \node[red dot] at (L2_M2_wu) {};
    \node[above=3pt of L2_M2_wu, red] {$m_{wu}^{(2)}$};

    \draw[red, dashed] (L2_M2_wu) -- (L3_M2_wu);

    \draw[red, thick, path arrow] (L3_M2_wu) -- (L3_U) node[midway, below=2pt] {$b$};

    \node[red dot] at (L3_M2_wu) {};
    \node[below=3pt of L3_M2_wu, red] {$m_{wu}^{(2)}$};

    \node[red dot big] at (L3_U) {};
    \node[above=3pt of L3_U, red] {$u$};

\end{tikzpicture}
\caption{Figure for Case 2ba of Case (i) in Theorem~\ref{idempotent hardness with certain conditions}.}
\label{fig:Case (i) Case 2ba}
\end{figure}

\item[Case 2bb: At $m_{wu}^{(2)}$, just after the second wrong edge, $P$ takes the edge $\{m_{wu}^{(2)},m_{wu}^{(1)}\}$]:\\
Then, the consecutive edges $\{u,m_{wu}^{(2)}\}$ and $\{m_{wu}^{(2)},m_{wu}^{(1)}\}$ in $P$ (labelled $b$ and $ba$ respectively) contribute a factor of $b\cdot (ba)$ (which is $0$ as $b^2=0$) in the expression of $P$'s yield, thus making it $0$; again, this is not possible as $P$'s yield is $\mu\neq 0$ instead (Refer Fig. \ref{fig:Case (i) Case 2bb}).

\begin{figure}[htbp]
\centering
\begin{tikzpicture}[
    scale=1.5,
    black dot/.style={circle, fill=black, inner sep=1.2pt},
    black dot big/.style={circle, fill=black, inner sep=2.2pt},
    red dot/.style={circle, fill=red, inner sep=1.2pt},
    red dot big/.style={circle, fill=red, inner sep=2.2pt},
    black edge/.style={black, thick},
    path arrow/.style={
        postaction={decorate, decoration={markings, mark=at position 0.55 with {\arrow[scale=1.5, red]{>}}}}
    }
]

    \coordinate (W)      at (0, 0);
    \coordinate (M1_wu) at (1.5, 0);
    \coordinate (M2_wu) at (3.0, 0);
    \coordinate (U)     at (4.5, 0);
    \coordinate (M1_uv) at (6.0, 0);
    \coordinate (M2_uv) at (7.5, 0);
    \coordinate (V)      at (9.0, 0);

    \draw[black edge] (W) -- (M1_wu) node[midway, below=4pt] {$a$};
    \draw[black edge] (M1_wu) -- (M2_wu) node[midway, below=4pt] {$ba$};
    \draw[black edge] (M2_wu) -- (U) node[midway, below=4pt] {$b$};
    \draw[black edge] (U) -- (M1_uv) node[midway, below=4pt] {$a$};
    \draw[black edge] (M1_uv) -- (M2_uv) node[midway, below=4pt] {$ba$};
    \draw[black edge] (M2_uv) -- (V) node[midway, below=4pt] {$b$};

    \node[black dot] at (W) {};
    \node[above=6pt of W] {$w$};

    \node[black dot] at (M1_wu) {};
    \node[above=6pt of M1_wu] {$m_{wu}^{(1)}$};

    \node[black dot] at (M2_wu) {};
    \node[above=6pt of M2_wu] {$m_{wu}^{(2)}$};

    \node[black dot big] at (U) {};
    \node[above=6pt of U] {$u$};

    \node[black dot] at (M1_uv) {};
    \node[above=6pt of M1_uv] {$m_{uv}^{(1)}$};

    \node[black dot] at (M2_uv) {};
    \node[above=6pt of M2_uv] {$m_{uv}^{(2)}$};

    \node[black dot] at (V) {};
    \node[above=6pt of V] {$v$};

    \coordinate (L1_U)     at (4.5, -0.9);
    \coordinate (L1_M1_uv) at (6.0, -0.9);

    \coordinate (L2_M1_uv) at (6.0, -1.6);
    \coordinate (L2_U)     at (4.5, -1.6);
    \coordinate (L2_M2_wu) at (3.0, -1.6);
    \coordinate (L2_M1_wu) at (1.5, -1.6);

    \draw[red, thick, path arrow] (L1_U) -- (L1_M1_uv) node[midway, below=2pt] {$a$};

    \node[red dot big] at (L1_U) {};
    \node[above=3pt of L1_U, red] {$u$};

    \node[red dot] at (L1_M1_uv) {};
    \node[above=3pt of L1_M1_uv, red] {$m_{uv}^{(1)}$};

    \draw[red, dashed] (L1_M1_uv) -- (L2_M1_uv);

    \draw[red, thick, dashed, path arrow] (L2_M1_uv) -- (L2_U) node[midway, below=2pt] {$a$};
    \draw[red, thick, dashed, path arrow] (L2_U) -- (L2_M2_wu) node[midway, below=2pt] {$b$};
    \draw[red, thick, dashed, path arrow] (L2_M2_wu) -- (L2_M1_wu) node[midway, below=2pt] {$ba$};

    \node[red dot] at (L2_M1_uv) {};
    \node[right=3pt of L2_M1_uv, red] {$m_{uv}^{(1)}$};

    \node[red dot big] at (L2_U) {};
    \node[above=3pt of L2_U, red] {$u$};

    \node[red dot] at (L2_M2_wu) {};
    \node[above=3pt of L2_M2_wu, red] {$m_{wu}^{(2)}$};

    \node[red dot] at (L2_M1_wu) {};
    \node[above=3pt of L2_M1_wu, red] {$m_{wu}^{(1)}$};

\end{tikzpicture}
\caption{Figure for Case 2bb of Case (i) in Theorem~\ref{idempotent hardness with certain conditions}.}
\label{fig:Case (i) Case 2bb}
\end{figure}
\end{description}
\end{description}

\item[Case 3: The first wrong edge of $P$ is of the form $\{m_{uv}^{(2)},m_{uv}^{(1)}\}$ for some edge $(u,v)\in E(G)$]:\\ As $\{m_{uv}^{(2)}, m_{uv}^{(1)}\}$ (which is labelled $ba$) is the first wrong edge picked by $P$, the two edges just before the first wrong edge in $P$, being right edges, must be $\{u,m_{uv}^{(1)}\}$(which is labelled a) and $\{m_{uv}^{(1)}, m_{uv}^{(2)}\}$ (which is labelled $ba$) (Refer Fig. \ref{fig:Case (i) Case 3}). Also, the subwalk (say $Q$) of $P$ before these two right edges (from $s$ to $u$) has edges of labels $a, ba, b, a, ba, b,\ldots, a, ba,b$ (in that order); so, the yield of $Q$ is $a\cdot (ba)\cdot b\cdot a\cdot (ba)\cdot b\ldots a\cdot(ba)\cdot b$, which, being a power of $ab$, is $ab$ (as $ab$ is idempotent).

Now, construct a new $s$-$t$ walk (say $\tilde{P}$) from $P$ obtained by bypassing two of its consecutive edges, namely the right edge just before the first wrong edge (i.e., the edge $\{m_{uv}^{(1)},m_{uv}^{(2)}\}$ labelled $ba$) and the the first wrong edge (i.e., the edge $\{m_{uv}^{(2)},m_{uv}^{(1)}\}$ labelled $ba$). Note that the yield of the subwalk of $P$ before these two bypassed edges is $($yield of $Q)\times$ $($label $a$ of the right edge $\{u,m_{uv}^{(1)}\}$ just before the right edge just before the first wrong edge$)$, which is $(ab)\cdot a$. Then, the two bypassed edges cause this partial yield $aba$ to get multiplied by their labels, i.e., $ba$ and $ba$. So, the yield of the subwalk of $P$ just after traversing these two bypassed edges becomes $(aba)\cdot (ba)\cdot (ba)$, which is $aba$ (as $ab$ is idempotent). Therefore, the yield just before traversing these two bypassed edges is the same as the yield just after traversing them. So, $\tilde{P}$ has the same yield as $P$, i.e., $\mu$. However, $\tilde{P}$ also has strictly smaller length than $P$; this is not possible, as $P$ has minimum length amongst all $s$-$t$ walks in $G'$ of yield $\mu$.
    \begin{figure}[htbp]
\centering
\begin{tikzpicture}[
    scale=1.3,
    black dot/.style={circle, fill=black, inner sep=1.2pt},
    black dot big/.style={circle, fill=black, inner sep=2.2pt},
    red dot/.style={circle, fill=red, inner sep=1.2pt},
    red dot big/.style={circle, fill=red, inner sep=2.2pt},
    black edge/.style={black, thick},
    path arrow/.style={
        postaction={decorate, decoration={markings, mark=at position 0.55 with {\arrow[scale=1.5, red]{>}}}}
    }
]

    \coordinate (S)    at (-6.0, 0);
    \coordinate (S1)   at (-5.1, 0);
    \coordinate (S2)   at (-4.2, 0);
    \coordinate (S3)   at (-3.3, 0);

    \coordinate (U_m3) at (-2.7, 0);
    \coordinate (U_m2) at (-1.8, 0);
    \coordinate (U_m1) at (-0.9, 0);

    \coordinate (U)     at (0, 0);
    \coordinate (M1_uv) at (1.5, 0);
    \coordinate (M2_uv) at (3.0, 0);
    \coordinate (V)      at (4.5, 0);

    \draw[black edge] (S)  -- (S1) node[midway, below=4pt] {$a$};
    \draw[black edge] (S1) -- (S2) node[midway, below=4pt] {$ba$};
    \draw[black edge] (S2) -- (S3) node[midway, below=4pt] {$b$};

    \node at (-3.0, 0) {$\dots$};

    \draw[black edge] (U_m3) -- (U_m2) node[midway, below=4pt] {$a$};
    \draw[black edge] (U_m2) -- (U_m1) node[midway, below=4pt] {$ba$};
    \draw[black edge] (U_m1) -- (U)    node[midway, below=4pt] {$b$};

    \draw[black edge] (U)     -- (M1_uv) node[midway, below=4pt] {$a$};
    \draw[black edge] (M1_uv) -- (M2_uv) node[midway, below=4pt] {$ba$};
    \draw[black edge] (M2_uv) -- (V)     node[midway, below=4pt] {$b$};

    \node[black dot] at (S) {};
    \node[above=6pt of S] {$s$};

    \node[black dot] at (S1) {};
    \node[black dot] at (S2) {};
    \node[black dot] at (S3) {};

    \node[black dot] at (U_m3) {};
    \node[black dot] at (U_m2) {};
    \node[black dot] at (U_m1) {};

    \node[black dot big] at (U) {};
    \node[above=6pt of U] {$u$};

    \node[black dot] at (M1_uv) {};
    \node[above=6pt of M1_uv] {$m_{uv}^{(1)}$};

    \node[black dot] at (M2_uv) {};
    \node[above=6pt of M2_uv] {$m_{uv}^{(2)}$};

    \node[black dot] at (V) {};
    \node[above=6pt of V] {$v$};

    \coordinate (L1_U)     at (0.0, -0.9);
    \coordinate (L1_M1_uv) at (1.5, -0.9);
    \coordinate (L1_M2_uv) at (3.0, -0.9);

    \coordinate (L2_M2_uv) at (3.0, -1.6);
    \coordinate (L2_M1_uv) at (1.5, -1.6);
    \coordinate (L2_U)     at (0.0, -1.6);

    \draw[red, thick, path arrow] (L1_U) -- (L1_M1_uv) node[midway, below=2pt] {$a$};
    \draw[red, thick, path arrow] (L1_M1_uv) -- (L1_M2_uv) node[midway, below=2pt] {$ba$};

    \node[red dot big] at (L1_U) {};
    \node[above=3pt of L1_U, red] {$u$};

    \node[red dot] at (L1_M1_uv) {};
    \node[above=3pt of L1_M1_uv, red] {$m_{uv}^{(1)}$};

    \node[red dot] at (L1_M2_uv) {};
    \node[above=3pt of L1_M2_uv, red] {$m_{uv}^{(2)}$};

    \draw[red, dashed] (L1_M2_uv) -- (L2_M2_uv);

    \draw[red, thick, dashed, path arrow] (L2_M2_uv) -- (L2_M1_uv) node[midway, below=2pt] {$ba$};

    \node[red dot] at (L2_M2_uv) {};
    \node[below=3pt of L2_M2_uv, red] {$m_{uv}^{(2)}$};

    \node[red dot] at (L2_M1_uv) {};
    \node[below=3pt of L2_M1_uv, red] {$m_{uv}^{(1)}$};

\end{tikzpicture}
\caption{Figure for Case 3 of Case (i) in Theorem~\ref{idempotent hardness with certain conditions}.}
\label{fig:Case (i) Case 3}
\end{figure}
    \end{description}
This concludes the argument for Case (i), i.e., $b^2=0$ and $a$ is idempotent. 

\paragraph{\underline{\textbf{Case (ii): $b^2=a^2=0$.}}}
The reduction from Case (i) works for Case (ii) too, except the following change in analysis of Case 2: Instead of making subcases, we get a contradiction right away as follows: The right edge just before the first wrong edge (i.e., $\{u,m_{uv}^{(1)}\}$ with label $a$) and the first wrong edge (i.e., $\{m_{uv}^{(1)},u\}$ with label $a$) together contribute a factor of $a^2=0$ to the yield of $P$, thus forcing it to be $0$; this is not possible as yield of $P$ is $\mu\neq 0$ instead. The analysis of Cases 1 and 3 from Case (i) works as is.

\paragraph{\underline{Case (iii): $a^2=0$ and $b$ is idempotent:}}
The case $a^2=0$ and $b^3=b$ is similar to Case (i), but with a few minor changes. Upon subdividing any edge, instead of setting the labels $a$, $ba$ and $b$, we set labels $b$, $ab$ and $a$ (in that order); that is, for every edge $(u,v)\in E(G)$,  $\phi(\{u,m_{uv}^{(1)}\}) = b$, $\phi(\{m_{uv}^{(1)},m_{uv}^{(2)}\}) = ab$ and $\phi(\{m_{uv}^{(2)}, v\}) = a$. 
Also, in the reduced graph $G'$, we add two extra vertices, say $s'$ and $t'$, which serve as a new source and a new sink, respectively. We also add an edge joining $s$ (i.e., the original source) and $s'$ with label $a$, and an edge joining $t$ (i.e., the original sink) and $t'$ with label $b$. 

It remains to argue that $G$ has a directed $s$-$t$ path if and only if $G'$ has an $s'$-$t'$ walk of yield $\mu=ab$. If $G$ has an $s$-to-$t$ path $P$, the corresponding $s$-to-$t$ path (let length($P$) = $\ell$) in $G'$ has yield $(b\cdot a)^{2\ell}$; also, prepending and appending this path with the edges $\{s',s\}$ and $\{t,t'\}$ respectively, we get an $s'$-to-$t'$ path in $G'$ of yield $a\cdot (b\cdot a)^{2 \ell}\cdot b = (a\cdot b)^{2 \ell+1} = \mu^{2 \ell+1}$, which is $\mu$ because $\mu$ is an idempotent element. 

Next, we argue the reverse direction; that is, given an $s'$-to-$t'$ walk $P$ of yield $\mu$ in $G'$, we show that $G$ has a directed $s$-to-$t$ path. Also, as in Case (i), let $P$ be chosen as a minimum length walk amongst all $s'$-$t'$ walks of yield $\mu$  in $G'$. Note that $\{s',s\}$ must be the first edge of $P$. Also, once $P$ has reached the vertex $s$ via the first edge, $P$ will never revisit $s'$. This is because if $P$ does visit $s'$, it will do so by travelling from $s$ to $s'$ via the edge $\{s,s'\}$, and then it will have to take the same edge to get back to $s$ (it cannot stay stuck at $s'$ as it has to eventually reach $t'$); so, this causes yield of $P$ to have $\big(\phi(\{s',s\})\big)^2 = a^2 = 0$ as a factor, which forces the overall yield to become $0$, which is not possible as it is supposed to be $\mu$. 

Now, let $P^{*}$ denote the subwalk of $P$ between $s$ and the first occurrence of $t$; as we just argued, this subwalk never revisits $s'$ and so, $P^{*}$ is completely contained in $G'\setminus \{s',t'\}$. For contradiction, assume that $G$ has no directed $s$-$t$ path. Then, the subwalk $P^{*}$ of $P$ does not translate back to a directed $s$-$t$ path in $G$. So, as before, consider the first wrong edge in $P^{*}$. Again, it can be of the three forms discussed earlier, i.e. 
$\{v,m_{uv}^{(2)}\}$ (Case 1) or $\{m_{uv}^{(1)},u\}$ (Case 2) or $\{m_{uv}^{(2)}, m_{uv}^{(1)}\}$ (Case 3)  for some edge $(u, v) \in E(G)$. The analysis of these cases is very similar to Case (i).

\begin{description}
\item[Case 1: The first wrong edge of $P^{*}$ is of the form $\{v,m_{uv}^{(2)}\}$ for some edge $(u,v)\in E(G)$]:\\
Like Case 1 of Case (i), the edge just before the first wrong edge  (in the walk $P^{*}$ in $G'$) must be $\{m_{wv}^{(2)}, v\}$ for some  $(w, v)\in E(G)$. Now, both the first wrong edge and the right edge (picked by $P^{*}$ just before the first wrong edge) have label $a$. So, the yield of $P^{*}$ (and so, also $P$) has $a^2=0$ as a factor, thus forcing it to be $0$; this is not possible as $P$'s yield is $\mu\neq 0$.

\item[Case 2: The first wrong edge of $P^{*}$ is of the form $\{m_{uv}^{(1)},u\}$ for some edge $(u,v)\in E(G)$]:\\
As $\{m_{uv}^{(1)}, u\}$ is the first wrong edge picked, the edge just before it (in the walk $P^{*}$ in $G'$), being a right edge, must be $\{u, m_{uv}^{(1)}\}$. Now again, as in Case 2 of Case (i), make two subcases based on whether the edge just after the first wrong edge is a right edge or a wrong edge. 
    \begin{description}
        \item[Case 2a:  At $u$, just after the first wrong edge, the next edge taken by $P^{*}$ is a right edge]:\\
        Then, this right edge (just after the first wrong edge) must be of the form $\{u,m_{uw}^{(1)}\}$ (which is labelled $b$) for some edge $(u,w) \in E(G)$. Now, construct a new $s$-$t$ walk $\tilde{P}$ from $P^{\star}$ obtained by bypassing two of its consecutive edges, namely the right edge just before the first wrong edge (i.e., the edge $\{u,m_{uv}^{(1)}\}$ labelled $b$), and the first wrong edge (i.e., the edge $\{m_{uv}^{(1)},u\}$ labelled $b$). Note that the edge $\{u,m_{uw}^{(1)}\}$ just after these two bypassed edges in $P$ has label $b$; so, after bypassing, in the expression of walk's yield, the factor $b\cdot b\cdot b$ (contributed by the edges $\{u,m_{uv}^{(1)}\}, \{m_{uv}^{(1)},u\}, \{u,m_{uw}^{(1)}\}$ in $P^{*}$) has been replaced by just a single $b$ (contributed by the edge $\{u,m_{uw}^{(1)}\}$ in $\tilde{P}$). So, as $b^3=b$ (since $b$ is idempotent), yield of $\tilde{P}$ is same as that of $P^{*}$. Then, prepending and appending $\tilde{P}$ with the edges $\{s',s\}$ and $\{t,t'\}$ respectively gives an $s'$-$t'$ walk in $G'$ whose yield is same as that of $P$ (i.e., $\mu$), and it has strictly smaller length than $P$; however, this is not possible because $P$ was chosen to have minimum length amongst all $s'$-$t'$ walks in $G'$ of yield $\mu$, a contradiction.
        
\item[Case 2b: At $u$, just after the first wrong edge, the next edge taken by $P^{*}$ is a wrong edge]:\\
Then, this second wrong edge of $P^{*}$ must be of the form $\{u,m_{wu}^{(2)}\}$ (labelled $a$) for some edge $(w,u) \in E(G)$. Again, as in Case 2b of Case (i), make two subsubcases based on which of the two edges incident to $m_{wu}^{(2)}$ does $P^{*}$ take just after this second wrong edge.
\begin{description}
\item[Case 2ba: At $m_{wu}^{(2)}$, just after the second wrong edge, $P^{*}$ takes the edge $\{m_{wu}^{(2)},u\}$]:\\
Then, the consecutive edges $\{u,m_{wu}^{(2)}\}$ and $\{m_{wu}^{(2)},u\}$ in $P^{*}$ (both are labelled $a$) contribute a factor of $a^2=0$ in the expression of walk $P^{*}$'s yield, thus making the yield of $P^{*}$ (and so, also $P$) to be $0$; this is not possible as the yield of $P$ is $\mu\neq 0$ instead.
\item[Case 2bb: At $m_{wu}^{(2)}$, just after the second wrong edge, $P^{*}$ takes the edge $\{m_{wu}^{(2)},m_{wu}^{(1)}\}$]:\\
Then, the consecutive edges $\{u,m_{wu}^{(2)}\}$ and $\{m_{wu}^{(2)},m_{wu}^{(1)}\}$ in $P^{*}$ (labelled $a$ and $ab$ respectively) contribute a factor of $a\cdot (ab)$ (which is $0$ as $a^2=0$) in the expression of $P^{*}$'s (and so, also $P$'s) yield, thus making it $0$; again, this is not possible as $P$'s yield is $\mu\neq 0$ instead.
\end{description}
\end{description}

\item[Case 3: The first wrong edge of $P^{*}$ is of the form $\{m_{uv}^{(2)},m_{uv}^{(1)}\}$ for some edge $(u,v)\in E(G)$]:\\ As $\{m_{uv}^{(2)}, m_{uv}^{(1)}\}$ (which is labelled $ab$) is the first wrong edge picked by $P^{*}$, the two edges just before the first wrong edge in $P^{*}$, being right edges, must be $\{u,m_{uv}^{(1)}\}$(which is labelled $b$) and $\{m_{uv}^{(1)}, m_{uv}^{(2)}\}$ (which is labelled $ab$). Also, the subwalk (say $Q$) of $P^{*}$ before these two right edges (from $s$ to $u$) has edges of labels $b, ab, a, b, ab, a,\ldots, b, ab,a$ (in that order); so, the yield of $Q$ is $b\cdot (ab)\cdot a\cdot b\cdot (ab)\cdot a\ldots b\cdot(ab)\cdot a$, which is $b$ and $a$ pre-multiplied and post-multiplied to a power of $ab$ respectively and this, in turn, becomes $b\cdot(ab)\cdot a$ (as $ab$ is idempotent).

Now, construct a new $s$-$t$ walk (say $\tilde{P}$) from $P^{*}$ obtained by bypassing two of its consecutive edges, namely the right edge just before the first wrong edge (i.e., the edge $\{m_{uv}^{(1)},m_{uv}^{(2)}\}$ labelled $ab$) and the first wrong edge (i.e., the edge $\{m_{uv}^{(2)},m_{uv}^{(1)}\}$ labelled $ab$). Note that the yield of the subwalk of $P$ before these two bypassed edges is $($yield of $Q)\times$ $($label $b$ of the right edge $\{u,m_{uv}^{(1)}\}$ just before the right edge just before the first wrong edge$)$, which is $(b\cdot (ab)\cdot a)\cdot b$; this, in turn, becomes $bab$ (as $ab$ is idempotent). Then, the two bypassed edges cause this partial yield $bab$ to get multiplied with their labels, i.e., $ab$ and $ab$. So, the yield of the subwalk of $P^{*}$ just after traversing these two bypassed edges becomes $(bab)\cdot (ab)\cdot (ab)$, which is $bab$ (as $ab$ is idempotent). Therefore, the yield just before traversing these two bypassed edges is the same as the yield just after traversing them. So, $\tilde{P}$ has the same yield as $P^{*}$. So then, appending and prepending the edges $\{s',s\}$ and $\{t,t'\}$ to $P^{*}$ gives an $s'$-$t'$ walk of yield $\mu$; however, this walk also has strictly smaller length than $P$, which is not possible as $P$ was chosen to be of minimum length amongst all $s'$-$t'$ walks of yield $\mu$ in $G'$.
\end{description}

This concludes the argument for Case (iii), i.e., $a^2=0$ and $b$ is idempotent. 
\end{proof}

As pointed out by a reviewer of the conference version of this paper earlier, the above proof works even if the condition $a^2=0$ (resp. $b^2=0$) is relaxed to saying that $a^2$ (resp. $b^2$) does not divide $\mu$ in $\langle a,b\rangle$ (i.e., monoid generated by $a$ and $b$); that is, $\lambda\cdot a^2\cdot \nu \neq \mu$ (resp. $\lambda\cdot b^2\cdot \nu \neq \mu$) for all $\lambda, \nu\in \langle a,b\rangle$. This is because, in the reverse direction of the above proof, the desired contradiction follows even if having $a^2$ (resp. $b^2$) as a factor forces the yield of the walk to become different from $\mu$ (not necessarily $0$). This gives us the following more general form of Theorem \ref{idempotent hardness with certain conditions}:

\begin{theorem}
\label{thm: relaxed premise idempotent hardness}
Consider any fixed finite monoid $M$. Let $\mu$ be an idempotent element of $M$ that can be factored as $\mu=a\cdot b$ for some $a,b\in M$ such that one of the following holds true:\\[-8mm]
\begin{enumerate}[(i)]
\itemsep -4pt
    \item $\mu\notin \langle a,b\rangle \cdot b^2 \cdot \langle a,b\rangle$ and $a$ is idempotent.
    \item $\mu\notin \langle a,b\rangle \cdot b^2 \cdot \langle a,b\rangle$ and $\mu\notin \langle a,b\rangle \cdot a^2 \cdot \langle a,b\rangle$.
    \item $\mu\notin \langle a,b\rangle \cdot a^2 \cdot \langle a,b\rangle$ and $b$ is idempotent.
\end{enumerate}
\vspace{-2mm}
Then, $\reachmf$ with $F=\{\mu\}$ is $\NL$-hard.
\end{theorem}

\subsection{$\NL$-hardness for any Monoid outside $\DS$ with Singleton Accepting Set}

We now use Theorem \ref{thm: relaxed premise idempotent hardness} can be used to argue the following:
\begin{theorem}
\label{thm: monoid outside DS hardness for some idempotent}
Consider any fixed finite monoid $M\notin \DS$. Then, there exists an idempotent $\mu\in M$ such that $\reachmf$ is $\NL$-hard with $F=\{\mu\}$.
\end{theorem}

To deduce Theorem \ref{thm: monoid outside DS hardness for some idempotent} using Theorem \ref{thm: relaxed premise idempotent hardness}, it suffices to prove the following lemma:

\begin{lemma}
\label{lem:choice of mu, a, b}
Consider any fixed finite monoid $M\not\in \DS$. Then, there exists an idempotent $\mu \in M$ and $a,b\in M$ such that $\mu$ can be factored as  $\mu = a\cdot b$,  $\mu\notin \langle a,b\rangle \cdot b^2\cdot \langle a,b\rangle$, and either $a=a^2$ or  $\mu\notin \langle a,b\rangle \cdot a^2\cdot \langle a,b\rangle$; in other words, there exists a choice of $\mu,a,b\in M$ that satisfies the conditions of either Case (i) or Case (ii) from Theorem \ref{thm: relaxed premise idempotent hardness}.
\end{lemma}

We present below two proofs of Lemma \ref{lem:choice of mu, a, b} in the next two subsections.

\subsubsection{Hardness via Homomorphisms}
\label{subsec:Hardness via Homomorphisms}
We now present the first proof of Lemma \ref{lem:choice of mu, a, b}

As $M\notin \DS$, it is divisible by $\BA_2$ or $\U$ (see conditions 3(a) and 4(c) of Theorem 3 in \cite{shevrin1995theory}) and so, there exists a surjective homomorphism (say $\phi$) from some submonoid $M'$ of $M$ to $\BA_2$ (or $\U$). We make the following two cases based on whether $M$ is divisible by $\BA_2$ or $\U$:

\begin{description}
\item[Case 1: $M$ is divisible by $\BA_2$]:\\
Arbitrarily pick $\gamma \in \phi^{-1}(\alpha)$ and $\delta\in \phi^{-1}(\beta)$. As $M$ is finite, there exists an integer $k\geq 1$ such that $(\gamma \cdot \delta)^{k}$ is an idempotent (see Proposition 6.31 in \cite{Pin25}). We take $\mu:= (\gamma \cdot  \delta)^{k}$, $a:= (\gamma\cdot\delta)^{k-1}\cdot \gamma$ and $b:= \delta$. Clearly, $\mu=a\cdot b$, and it is idempotent. Let us show that the conditions of Case (ii) from Theorem \ref{thm: relaxed premise idempotent hardness} are met by this choice of $\mu, a,b$. First, we show that $b^2$ does not divide $\mu$ in $\langle a,b \rangle$. For contradiction, assume that $\mu = \lambda\cdot b^2\cdot \nu$ for some $\lambda, \nu\in \langle a,b\rangle$. Taking the image of both sides under the homomorphism $\phi$, we get $\phi(\mu) = \phi(\lambda)\cdot \phi(b^2)\cdot \phi(\nu)$. Note that $\phi(b^2) = (\phi(\delta))^2 = \beta^2=0$  and so, RHS $=0$; however, LHS $= \phi(\mu) = \phi((\gamma\cdot \delta)^{k}) = (\phi(\gamma)\cdot \phi(\delta))^{k} = (\alpha\beta)^{k}$, which is $\alpha\beta$ (because $\alpha\beta$ is idempotent), which is $\neq $ RHS, which is a contradiction, as desired. Next, let us show that $a^2$ does not divide $\mu$ in $\langle a,b\rangle$. For contradiction, assume that $\mu=\lambda\cdot a^2\cdot\nu$ for some $\lambda,\nu\in \langle a,b\rangle$. Again, taking the image of both sides under $\phi$, we get $\phi(\mu)=\phi(\lambda)\phi(a^2)\phi(\nu)$. Note that $\phi(a^2) = (\phi((\gamma\cdot\delta)^{k-1}\cdot\gamma))^2$, which has $(\phi(\gamma))^2=\alpha^2=0$ as a factor and so, RHS $=0$; however, LHS $=\phi(\mu)=\alpha\beta\neq0$, a contradiction.

\item[Case 2: $M$ is divisible by $\U$]:\\
Arbitrarily pick $\gamma \in \phi^{-1}(\alpha)$ and $\delta\in \phi^{-1}(\beta)$. As $M$ is finite, there exists an integer $\ell\geq 1$ such that $\gamma^{\ell}$ is an idempotent, and there exists an integer $k\geq 1$ such that $(\gamma^{\ell} \cdot \delta)^k$ is an idempotent. Also, following Proposition 6.31 in \cite{Pin25}, let us take a closer look at the proof showing the existence of such $k$. In the sequence of powers of $\gamma^{\ell}\delta$, consider the first repetition; say the $i^{th}$ power matches with $(i+p)^{th}$ power, i.e., 
$(\gamma^{\ell}\delta)^{i} = (\gamma^{\ell}\delta)^{i+p}$. Here, $i$ and $p$ are referred to as `index' and `period' of $\gamma^{\ell}\delta$, respectively. Then, the overall sequence is $(\gamma^{\ell}\delta)^{1}, \ldots, (\gamma^{\ell}\delta)^{i-1}$ followed by infinite repetitions of the $p$-length subsequence $(\gamma^{\ell}\delta)^{i},\ldots, (\gamma^{\ell}\delta)^{i+p-1}$. Here, for any integers $t_1,t_2$ such that $i\leq t_1<t_2$, the $t_1^{th}$ and $t_2^{th}$ powers match (i.e., $(\gamma^{\ell}\delta)^{t_1} = (\gamma^{\ell}\delta)^{t_2})$ if and only if their difference $t_2-t_1$ is a multiple of $p$. So then, as described in \cite{Pin25}, choosing $k$ as any multiple of $p$ (say $k=qp$, where $q$ is an integer) such that $qp\geq i$, we get that $(\gamma^{\ell}\delta)^k$ is an idempotent. For our purposes, we choose $q$ such that $(q-1)\cdot p\geq i$. Now, take $\mu:= (\gamma^{\ell} \cdot  \delta)^{k}$, $a:= (\gamma^{\ell}\cdot\delta)^{k-p}\cdot \gamma^{\ell}$ and $b:= \delta\cdot (\gamma^{\ell}\cdot\delta)^{p-1}$. Clearly, $\mu=a\cdot b$, and it is idempotent. We show that the conditions of Case (i) from Theorem \ref{thm: relaxed premise idempotent hardness} are met by this choice of $\mu, a,b$.

First, we show that
$b^2$ does not divide $\mu$ in $\langle a,b \rangle$. For contradiction, assume that $\mu = \lambda\cdot b^2\cdot \nu$ for some $\lambda, \nu\in \langle a,b\rangle$. Taking the image of both sides under the homomorphism $\phi$, we get $\phi(\mu) = \phi(\lambda)\cdot \phi(b^2)\cdot \phi(\nu)$. Note that $\phi(b^2) = \phi\big((\delta\cdot (\gamma^{\ell}\cdot\delta)^{p-1})^2\big)$ is $0$ because it has $(\phi(\delta))^2= \beta^2=0$ as a factor and so, RHS $=0$; however, LHS $= \phi(\mu) = \phi((\gamma^{\ell}\cdot \delta)^{k}) = (\alpha^{\ell}\beta)^{k}$,  which is $(\alpha\beta)^{k}$ (as $\alpha$ is idempotent), which is $\alpha\beta$ (as $\alpha\beta$ is idempotent), which is $\neq $ RHS, a contradiction, as desired. Next, let us show that $a$ is idempotent. Note that $a^2 = \big((\gamma^{\ell}\cdot\delta)^{k-p}\cdot \gamma^{\ell}\big)\cdot \big((\gamma^{\ell}\cdot\delta)^{k-p}\cdot \gamma^{\ell}\big)$, which is $(\gamma^{\ell}\cdot \delta)^{k-p}\cdot (\gamma^{\ell}\cdot\delta)^{k-p}\cdot \gamma^{\ell}$ (because $\gamma^{\ell}$ is idempotent), which equals $(\gamma^{\ell}\cdot\delta)^{2k-2p}\cdot \gamma^{\ell}$. Also, $(\gamma^{\ell}\cdot\delta)^{2k-2p} = (\gamma^{\ell}\cdot\delta)^{k-p}$ because their exponents $2k-2p$ and $k-p$ differ by $k-p$, which is a multiple of $p$ (i.e., $(q-1)\cdot p$) and both of them are $\geq i$. Therefore, $a^2 = (\gamma^{\ell}\cdot\delta)^{k-p}\cdot \gamma^{\ell}$, which is $a$, as desired.
\end{description}

This concludes the first proof of Lemma \ref{lem:choice of mu, a, b}.

\subsubsection{Hardness via Rees Matrix Representation}
\label{sec:reviewers-proof}

We now show an alternate proof of Lemma \ref{lem:choice of mu, a, b}. As $M\notin\DS$, by the definition of $\DS$, there exists a regular $\calD$-class $A \subseteq M$ that is not a subsemigroup of $M$. Let $M(A) = A \cup \{0\}$ be the principal factor (see Section \ref{ec:prelims}) associated with $A$. As $M(A)$ is completely 0-simple (see Lemma 2.4 in \cite{dandan2019group}), we apply Rees theorem; that is, $M(A)$ is isomorphic to a Rees matrix semigroup (see Section \ref{ec:prelims}) of the form $\mathcal{M}^0(G; I, \Lambda; P)$ for some group $G$. First, we briefly mention the following points that are explained in the proof of Rees theorem in  \cite{How95} (see Theorem 3.2.3 therein for more details), including the choice of the group $G$, sets $I,\Lambda$ (for indexing columns and rows of sandwich matrix $P$ respectively), entries of $P$ and the isomorphism from $\mathcal{M}^0(G; I, \Lambda; P)$ to $M(A)$:

\begin{enumerate}
\item 
\label{item: 2 type H class}
Any $\mathcal{H}$-class $H$ in the $\mathcal{D}$-class $M(A)\setminus \{0\}$ (i.e., $A$) of $M(A)$ is of one of two types: i) \emph{group} $\mathcal{H}$-\emph{class} and ii) \emph{zero} $\mathcal{H}$-\emph{class}; $H$ is of first type when $x\cdot y\in H$ for $x,y\in H$ (in this case, $H$ forms a group), and $H$ is of second type when  $x\cdot y = 0$ for $x,y\in H$ (in this case, $H^2=\{0\}$).

\item 
\label{item: at least one H class}
Every $\mathcal{R}$-class and every $\mathcal{L}$-class in the $\mathcal{D}$-class $A$ of $M(A)$ contains at least one group $\mathcal{H}$-class.

\item 
\label{item: R-L classes}
In $M(A)$, for any $x,y\in A$, either $x\cdot y = 0$ or $x\cdot y$ is contained in the intersection of the $\mathcal{R}$-class containing $x$  and the $\mathcal{L}$-class containing $y$; the latter occurs if and only if there exists an idempotent in the intersection of the $\mathcal{L}$-class containing $x$ and the $\mathcal{R}$-class containing $y$.

\item 
\label{item: H11 is a group}
$\Lambda$ (resp. $I$) consists of all non-zero $\mathcal{L}$-classes (resp. $\mathcal{R}$-classes) of $M(A)$. Without loss of generality, the first $\mathcal{R}$-class $R_1$ and the first $\mathcal{L}$-class $L_1$ intersect in a group $H$-class $H_{1,1}$ (chosen as the group $G$); otherwise, $\mathcal{L}$ and $\mathcal{R}$-classes can be reordered in $\Lambda$ and $I$ to ensure this.

\item 
\label{item: Pick item arbitrarily}
For each $1\leq i\leq |I|$, an element $r_i$ is arbitrarily picked from the intersection of  $i^{th}$ $\mathcal{R}$-class $R_i$ and first $\mathcal{L}$-class $L_1$ (i.e., from the $\mathcal{H}$-class $H_{i,1}$). Similarly, for each $1\leq \lambda\leq |\Lambda|$, an element $q_{\lambda}$ is arbitrarily picked from the intersection of $\lambda^{th}$ $\mathcal{L}$-class $L_{\lambda}$ and first $\mathcal{R}$-class $R_1$ (i.e., from the $\mathcal{H}$-class $H_{1,\lambda}$). 

\item 
\label{item: Matrix entries}
For each $1\leq i\leq |I|$ and each $1\leq \lambda\leq |\Lambda|$, the matrix entry $P_{\lambda, i}$ is set as $q_{\lambda}\cdot r_{i}$ when $H_{i,\lambda}$ is a group $\mathcal{H}$-class; otherwise (i.e., when $H_{i,\lambda}$ is a zero $\mathcal{H}$-class), the matrix entry $P_{\lambda, i}$ is set as 0.

\item 
\label{item: isomorphism}
$\mathcal{M}^0(G;I,\Lambda;P)$ is isomorphic to $M(A)$ via the isomorphism $\phi: \mathcal{M}^0(G;I,\Lambda;P)\rightarrow M(A)$ defined as follows: For every $(i,g,\lambda)\in I\times G\times \Lambda$, $\phi(i,g,\lambda):= r_i\cdot g\cdot q_{\lambda}$, and $\phi(0):=0$.
\end{enumerate}
Next, for our purpose, we further reorder $\mathcal{L}$-classes in $\Lambda$ and $\mathcal{R}$-classes in $I$ such that, in addition to $H_{1,1}$ being a group $\mathcal{H}$-class (as in point \ref{item: H11 is a group} above), it is also ensured that the intersection of second $\mathcal{R}$-class $R_2$ with (i) first $\mathcal{L}$-class $L_1$ is a zero $\mathcal{H}$-class $H_{2,1}$ and (ii) second $\mathcal{L}$-class $L_2$ is a group $\mathcal{H}$-class $H_{2,2}$. Such a reordering can be done as follows:

\begin{itemize}
\item As $A$ does not form a subsemigroup of $M$, there exist $x,y\in A$ such that $x\cdot y\notin A$; so, $x\cdot y$ is not   $\mathcal{J}$-related to $x$. Also, as $M(x\cdot y)M\subseteq MxM$, $x\cdot y \leq_{\mathcal{J}} x$. Therefore, $x\cdot y <_{\mathcal{J}} x$
and so, $x\cdot y\in I_A$ in $M$. Also, in the principal factor $M(A)$ associated with $A$, $I_A$ is identified with $0$. Therefore, $x\cdot y = 0$ in $M(A)$. So, as per  point \ref{item: R-L classes} above, it follows that $L_{\lambda}\cap R_i$ (i.e., the $\mathcal{H}$-class $H_{i,\lambda}$) does not contain an idempotent, where  $L_{\lambda}$ and $R_i$ denote the $\mathcal{L}$-class containing $x$ and $\mathcal{R}$-class containing $y$ respectively. So, $H_{i,\lambda}$ cannot be a group and thus, as per point \ref{item: 2 type H class} above, it must be a zero $\mathcal{H}$-class. Swap second $\mathcal{R}$-class $R_2$ with $i^{th}$ $\mathcal{R}$-class $R_i$, and swap first $\mathcal{L}$-class $L_1$ with $\lambda^{th}$ $\mathcal{L}$-class $L_{\lambda}$. Now, $H_{2,1}$ becomes a zero $\mathcal{H}$-class.

\item If $H_{1,1}$ is not already a group $\mathcal{H}$-class, consider $i>2$ such that $H_{i,1}$ is a group $\mathcal{H}$-class; such $i$ exists because the first $\mathcal{L}$-class $L_1$ has at least one group $\mathcal{H}$-class (see point \ref{item: at least one H class} above). Swap first $\mathcal{R}$-class $R_1$ with $i^{th}$ $\mathcal{R}$-class $R_i$. Now, $H_{1,1}$ becomes a group $\mathcal{H}$-class.
\item If $H_{2,2}$ is not already a group $\mathcal{H}$-class, consider $\lambda>2$ such that $H_{2,\lambda}$ is a group $\mathcal{H}$-class; such $\lambda$ exists because the second $\mathcal{R}$-class $R_2$ has at least one group $\mathcal{H}$-class (see point \ref{item: at least one H class} above). Swap second $\mathcal{L}$-class $L_2$ with $\lambda^{th}$ $\mathcal{L}$-class $L_{\lambda}$. Now, $H_{2,2}$ becomes a group $\mathcal{H}$-class.
\end{itemize}

We make two cases based on whether $H_{1,2}$ is a zero $\calH$-class or a group $\calH$-class.

\begin{description}\item[Case 1: $H_{1,2}$ is a zero $\calH$-class]:\\
Let $\gamma:=\phi(1,e,2)$ and $\delta:=\phi(2,e,1)$. As $H_{2,2}$ is a group $\calH$-class, $P_{2,2}$ is set to $q_2\cdot r_2\in H_{1,1}$ (as per point \ref{item: Matrix entries} above). So, $\gamma\cdot \delta = \phi(1,e\cdot P_{2,2}\cdot e,1)$ becomes $\phi(1,q_2\cdot r_2,1)$, which becomes $r_1\cdot (q_2\cdot r_2)\cdot q_1$ (as per  point \ref{item: isomorphism} above); here, as each of $r_1, q_2\cdot r_2, q_1$ individually belongs to the group $\calH$-class $H_{1,1}$, so does their product. That is, $\gamma\cdot\delta\in H_{1,1}$. Also, again, as $H_{1,1}$ is a group $\calH$-class, all powers of $\gamma\cdot\delta$ also belong to $H_{1,1}$; so, because of finiteness of $M$, there exists $k\geq 1$ such that $(\gamma\cdot \delta)^k\in H_{1,1}$ is an idempotent. Now, choose $\mu:= (\gamma\cdot \delta)^k$, $a:=(\gamma\cdot\delta)^{k-1}\cdot \gamma$, $b:=\delta$. Note that $\mu=a\cdot b$ and clearly, it is an idempotent in $H_{1,1}\subseteq A$. Let us show that the conditions of Case~(ii) from Theorem~\ref{thm: relaxed premise idempotent hardness} are met. To that end, let us show that $b^2<_{\calJ}\mu$ (and so, $\mu\not\in Mb^2M$) and $a^2<_{\calJ} \mu$ (and so, $\mu\not\in Ma^2M$) in $M$. This amounts to showing that $a^2=b^2=0$ in $M(A)$ because recall that in the principal factor associated with $A$, it is precisely the ideal formed by elements $<_{\calJ}$ $A$ (and so, also $\mu$, for $\mu\in A$) that is identified with $0$. First, let us show that $b^2=0$ in $M(A)$. Note that $b^2 = \delta^2 = \phi(2,e,1)\cdot \phi(2,e,1)$; as $H_{2,1}$ is a zero $\calH$-class, this becomes $0$ in $M(A)$, as desired. Next, let us show that $a^2=0$ in $M(A)$. Note that $a^2 = ((\gamma\cdot \delta)^{k-1}\cdot \gamma)\cdot((\gamma\cdot \delta)^{k-1}\cdot \gamma)$; this is $0$ in $M(A)$ as it contains $\gamma^2 = \phi(1,e,2)\cdot \phi(1,e,2)$ as a factor, which becomes $0$ as $H_{1,2}$ is a zero $\calH$-class, as desired.

\item[Case 2: $H_{1,2}$ is a group $\calH$-class]:\\
Let $\gamma:=\phi(1,e,2)$ and $\delta:=\phi(2,e,1)$. As $M$ is finite, there exist $\ell, k\geq 1$ such that $\gamma^{\ell}$ and $(\gamma^{\ell}\cdot \delta)^{k}$ are idempotents. Also, in a manner similar to Case~2 of Section~\ref{subsec:Hardness via Homomorphisms}, $k$ here is chosen as $qp$, where $p$ is the period of $\gamma^{\ell}\cdot \delta$, and $q$ is sufficiently large so as to ensure that $(q-1)\cdot p$ is at least the index $i$ of $\gamma^{\ell}\cdot \delta$. Then, choose $\mu:=(\gamma^{\ell}\cdot \delta)^k$, $a:= (\gamma^{\ell}\delta)^{k-p}\cdot \gamma^{\ell}$ and $b:=\delta\cdot (\gamma^{\ell}\cdot \delta)^{p-1}$. Clearly, $\mu=a\cdot b$ and it is idempotent. Note that $\gamma^{\ell}\cdot \delta = (\phi(1,e,2))^{\ell}\cdot \phi(2,e,1)$, which becomes $(r_1\cdot e \cdot q_2)^\ell\cdot (r_2\cdot e\cdot q_1)$ (see point \ref{item: isomorphism} above), which can be expressed as $r_1\cdot (q_2\cdot r_1)^{\ell-1}\cdot (q_2\cdot r_2)\cdot q_1$. Note that $q_2\cdot r_1$ belongs to the group $\calH$-class $\calR_{q_2}\cap \calL_{r_1} = H_{1,1}$ (because $\calR_{r_1}\cap \calL_{q_2} = H_{1,2}$ is a group $\calH$-class and so, it has an idempotent; see point \ref{item: R-L classes} above). Further, $q_2\cdot r_2$ also belongs to the group $\calH$-class $\calR_{q_2}\cap \calL_{r_2} = H_{1,1}$ (because $\calR_{r_2}\cap \calL_{q_2} = H_{2,2}$ is a group $\calH$-class and so, it has an idempotent; again, see point \ref{item: R-L classes} above). Thus, as each of $r_1, q_2\cdot r_1, q_2\cdot r_2, q_1$ belongs to the group $\calH$-class $H_{1,1}$, so does $\gamma^{\ell}\cdot \delta$ (since, as worked out above, it can be expressed as $r_1\cdot (q_2\cdot r_1)^{\ell-1}\cdot (q_2\cdot r_2)\cdot q_1$) and its $k^{th}$ power $\mu$. Hence, $\mu\in H_{1,1}\subseteq A$.

Next, let us show that the conditions of Case~(i) from Theorem~\ref{thm: relaxed premise idempotent hardness} are met. To that end, it suffices to show that $b^2=0$ in $M(A)$ and $a=a^2$. Note that $b^2=(\delta\cdot (\gamma^{\ell}\cdot \delta)^{p-1})\cdot (\delta\cdot (\gamma^{\ell}\cdot \delta)^{p-1})$ has $\delta^2 = \phi(2,e,1)\cdot \phi(2,e,1)$ as a factor, which is $0$ because $H_{2,1}$ is a zero $\calH$-class; so, $b^2$ is indeed $0$ in $M(A)$. Next, let us show that $a^2=a$. Note that $a^2 = \big((\gamma^{\ell}\cdot\delta)^{k-p}\cdot \gamma^{\ell}\big)\cdot \big((\gamma^{\ell}\cdot\delta)^{k-p}\cdot \gamma^{\ell}\big)$, which is $(\gamma^{\ell}\cdot \delta)^{k-p}\cdot (\gamma^{\ell}\cdot\delta)^{k-p}\cdot \gamma^{\ell}$ (because $\gamma^{\ell}$ is idempotent), which equals $(\gamma^{\ell}\cdot\delta)^{2k-2p}\cdot \gamma^{\ell}$. Also, $(\gamma^{\ell}\cdot\delta)^{2k-2p} = (\gamma^{\ell}\cdot\delta)^{k-p}$ because their exponents $2k-2p$ and $k-p$ differ by $k-p$, which is a multiple of $p$ (i.e., $(q-1)\cdot p$) and both of them are $\geq i$. Therefore, $a^2 = (\gamma^{\ell}\cdot\delta)^{k-p}\cdot \gamma^{\ell}$, which is $a$, as desired.

\end{description}
This concludes the second proof of Lemma \ref{lem:choice of mu, a, b}.

\section{Reachability with Arbitrary Accepting Sets}

The most interesting general case of this problem is when the accepting subset can be an arbitrary subset of the monoid $M$. Our results in this direction are (a) for all commutative monoids (Section~\ref{subsec: Logspace Algorithms for Commutative Monoids}), thus generalising the logspace algorithm for all commutative aperiodic monoids in \cite{RSS19},  (b) for certain subclasses of union-of-groups 
(Section~\ref{subsec: Logspace Algorithms for Union-of-Groups Monoids}), thus generalising the known algorithm for the case when the monoid is a group by itself, and (c) for the monoids $\BA_2$ and $\U$ (Section~\ref{subsec: Dichotomy for BA_2 and U with Different Accepting Sets}), which are the building blocks of non-commutative aperiodic monoids.


\subsection{Logspace Algorithms for Commutative Monoids}
\label{subsec: Logspace Algorithms for Commutative Monoids}
In the following theorem, we show a logspace algorithm for all commutative monoids with arbitrary accepting subsets.

\commutativeArbitraryAcceptance*

\begin{proof}
Let $\alpha_1, \ldots, \alpha_k$ denote the elements of $M$. Let us describe a logspace algorithm that solves $\textsc{Reach}_{M,\{\alpha\}}$ for any accepting element $\alpha\in M$. For any $1\leq i\leq k$, consider the sequence consisting of all powers of $\alpha_i$, i.e., $\alpha_i^{0}, \alpha_i^{1}, \alpha_i^{2}, \ldots$ and so on. As $M$ is finite, this sequence must have a repeated element; consider the least $j_i$ such that the $j_i^{th}$ term $\alpha_i^{j_i}$ is repetition of some term (say, $k_i^{th}$ term $\alpha_i^{k_i}$) amongst the first $j_i-1$ terms of the sequence. That is, $\alpha_i^{k_i} = \alpha_i^{j_i}$. Then, the sequence takes the form $\alpha_i^{0}, \alpha_i^{1}, \ldots, \alpha_i^{k_i-1}$ followed by repetitions of  $\alpha_i^{k_i}, \alpha_{i}^{k_i+1},\ldots, \alpha_i^{j_i-1}$. 

Because of commutativity, the yield of any $s$-$t$ walk in the input graph $G$ is uniquely determined by the number of occurrences (say $n_i$) of $\alpha_i$ as an edge label in the walk for each $1\leq i\leq k$. Also, for each $1\leq i\leq k$, if $n_i$ is $\geq k_i$; then what matters is only the modulo $(j_i-k_i)$ difference between $n_i$ and $k_i$. Now, we iterate over all tuples $$(n_1', \ldots, n_k')\in \underset{1\leq i\leq k}{\bigtimes} \{0,1 , \ldots, k_i-1, k_i, \ldots, j_i-1\}$$ such that the accepting element $\alpha$ can be expressed as $\alpha_1^{n_1'}\alpha_2^{n_2'}\ldots \alpha_k^{n_k'}$. There is only a constant number of such tuples. Then, for each such tuple $(n_1', \ldots, n_k')$, it remains to check whether there is an $s$-$t$ walk $P$ in $G$ such that for each $1\leq i\leq k$, (a) if $n_i'\leq k_i-1$, then $P$ contains exactly $n_i'$ edges labelled $\alpha_i$, and (b) if $k_i \leq n_i'\leq j_i-1$, then $P$ contains $k_i$ edges labelled $\alpha_i$ plus a few more edges labelled $\alpha_i$ whose count modulo $(j_i-k_i)$ is $n_i'-k_i$. Next, we guess (i.e., iterate over all possible choices of) these $n_i'$ edges $\big($from (a)) for all $1\leq i\leq k$ such that $n_i'\leq k_i-1\big)$ and $k_i$ edges $\big($from (b)) for all $1\leq i\leq k$ such that $k_i \leq n_i'\leq j_i-1\big)$ in the order in which they appear in the desired walk, say $(u_1, v_1), (u_2, v_2)\ldots, (u_N, v_N)$, where $$N:= \underset{\substack{1\leq i\leq k:\\ n_i'\leq k_i-1}}{\sum}n_i' + \underset{\substack{1\leq i\leq k:\\ n_i'\geq k_i}}{\sum}k_i$$ There are only polynomially many such guesses (and so, the counter needed to iterate over them uses only logarithmic space). 

Then, for each such guess, it remains to check whether there are $s-u_1$, $v_1-u_2$, $v_2-u_3$, $\ldots$, $v_{N-1}-u_N$, $v_N-t$ walks in $G$ such that these $N+1$ walks together consist of $n_i'-k_i$ edges modulo $(j_i-k_i)$ labelled $\alpha_i$ for each $1\leq i\leq k$ such that $n_i'\geq k$. We further guess how these counts modulo $(j_i-k_i)$ are distributed across the $N+1$ walks. That is, for each $1\leq i\leq k$ such that $n_i'\geq k$, we iterate over all $\mu_{i,1}, \mu_{i,2}, \ldots, \mu_{i, N+1}\in \{0,1, \ldots, j_i-k_i-1\}$ such that $\sum_{\lambda=1}^{N+1}\mu_{i,\lambda} \equiv n_i'-k_i ~(mod~(j_i-k_i))$. There is only a constant number of such guesses. Then, for each such guess, it remains to check for each $1\leq \lambda\leq N+1$, whether there is a $v_{\lambda-1}-u_{\lambda}$ walk in $G$ that consists of $\mu_{i,\lambda}$ many edges (modulo $j_i-k_i$) labelled $\alpha_i$ for each $1\leq i\leq k$ such that $n_i'\geq k$. Let $\alpha_{i_1}, \alpha_{i_2}, \ldots, \alpha_{i_r}$ denote all $\alpha_i$'s of the latter kind. Discard all edges of $G$ whose labels are different from any of these $r$ elements. Also, for each remaining edge, replace its label $\alpha_{i_q}$ (where $1\leq q\leq r$) with a $r$-tuple $(0,0,\ldots, 1, \ldots,0, 0)$ from the group $H:=\mathbb{Z}_{j_{i_1}-k_{i_1}}\times \ldots \times \mathbb{Z}_{j_{i_r}-k_{i_r}}$, where 1 is placed at the position $q$ in the $r$-tuple. In the resulting graph, the task now reduces to checking for each $1\leq \lambda\leq N+1$, whether there is a $v_{\lambda-1}-u_{\lambda}$  walk whose yield is $(\mu_{i_1,\lambda}, \ldots, \mu_{i_r, \lambda})$; to do so, use the logspace algorithm \cite{RSS19} over group $H$. This proves Theorem~\ref{thm: logspace algo commutative monoids}. 
\end{proof}

\subsection{Logspace Algorithms for Union-of-Groups Monoids}
\label{subsec: Logspace Algorithms for Union-of-Groups Monoids}
Recall the definition of $\calL$-commutative and $\calR$-commutative monoids from the preliminaries. In the theorem below, we show deterministic logspace algorithms for any accepting subset of such union-of-groups monoids.

\subsubsection{$\calL$- (or $\mathcal{R}$-) Commutative Union-of-Groups Monoids}
\label{subsubsec: Over L- (or R-) Commutative Union-of-Groups Monoids}
\firstlcommutativealgo*
\begin{proof}
Assume that $M$ is $\calL$-commutative; the proof for the case when $M$ is $\calR$-commutative would be similar. Let $G'$ denote the product graph of the input graph $G(V, E)$.  Note that $V(G')$ can be partitioned into parts $V\times H$ for all $\mathcal{H}$-classes $H$ of $M$. We use the following property of $\calL$-commutative (and respectively $\calR$-commutative) monoids.

\begin{lemma}[\cite{Gig20}]
\label{L-commutativity same as R-classes coinciding with H-classes lemma}
For any $\calL$-commutative (resp. $\calR$-commutative) monoid $M$, the $\mathcal{R}$-classes (resp. $\mathcal{L}$-classes)  must coincide with its $\mathcal{H}$-classes. 
\end{lemma}
\begin{proof}
We present a proof for completeness, and only for the case when $M$ is $\calL$-commutative (the proof for $\calR$-commutative is similar). Consider any $a, b\in M$ from the same $\mathcal{R}$-class. We argue that $a$ and $b$ must also be from the same $\mathcal{H}$-class. As $a$ and $b$ are in the same $\mathcal{R}$-class, there exist $c, d\in M$ such that $a\cdot c= b$ and $b\cdot d = a$. Now, as $M$ is $\calL$-commutative, both $ac$ (same as $b$) and $ca$ belong to the same $\calL$-class; so, there exists $e\in M$ such that $e \cdot ca = b$. So, as $ec \cdot a = b$, we have $b \leq _{\calL} a $.  Similarly, as $M$ is $\calL$-commutative, both $bd$ (same as $a$) and $db$ belong to the same $\calL$-class; so, there exists $f\in M$ such that $f \cdot db = a$. So, as $fd\cdot  b =a$, we have $a \leq _{\calL} b$. Therefore, we have $a \calL b$; but then, as $a$ and $b$ also belong to the same $\mathcal{R}$-class, this gives $a \mathcal{H} b$, as desired. This proves Lemma~\ref{L-commutativity same as R-classes coinciding with H-classes lemma}. 
\end{proof}

The rest of the proof is in two steps. We first prove that, for any $\mathcal{H}$-class $H$, once a path $P$ in $G'$ exits the part $V\times H$, $P$ cannot re-visit $V\times H$. We then show that $G'[V \times H]$ is an Eulerian subgraph of $G'$, and hence the reachability within the subgraph can be tested in deterministic logspace.


\begin{lemma}
\label{no re-visit lemma}
Once a path $P$ in $G'$ exits the part $V\times H$ for any $\mathcal{H}$-class $H$, $P$ cannot re-visit $V\times H$. 
\end{lemma}

\begin{proof}
    For the sake of contradiction, assume that $P$ exits the part $V\times H$ (say, at vertex $(v,h)$), enters the part $V\times \widetilde{H}$ for another $\mathcal{H}$-class $\widetilde{H}$ (say, at vertex $(w,\widetilde{h})$), and then eventually re-enters the part $V\times H$ (say, at vertex $(v', h')$) (see fig.~\ref{fig:path revisiting}). Now, as $(w, \widetilde{h})$ is reachable from $(v, h)$ in $G'$, there exists an $m\in M$ such that $h\cdot m = \widetilde{h}$ and so, $\widetilde{h}\leq _{\mathcal{R}} h$. Also, as $(v', h')$ is reachable from $(w, \widetilde{h})$ in $G'$, there exists an $m'\in M$ such that $\widetilde{h}\cdot m' = h'$; so, $\widetilde{h}\cdot (m'\cdot ((h')^{-1}\cdot h)) = h$ and thus, $h\leq _{\mathcal{R}} \widetilde{h}$. \\

\begin{figure}
    \centering
    \begin{tikzpicture}
  
  \definecolor{mybg}{RGB}{255,255,255}

  \filldraw[fill=mybg, draw=black, drop shadow={opacity=0.25, shadow xshift=2pt, shadow yshift=-2pt}] 
    (0, 2.5) ellipse (4cm and 1cm);
  \filldraw[fill=mybg, draw=black, drop shadow={opacity=0.25, shadow xshift=2pt, shadow yshift=-2pt}] 
    (0, 0) ellipse (4cm and 1cm);

  \node at (3, 2.5) {$V \times H$};
  \node at (3, 0) {$V \times \widetilde{H}$};

  \node at (0, 3.8) {$G'$};

  \node (v_h) at (-2, 2.5) {$(v,h)$};
  \node (v_h_prime) at (1.5, 2.5) {$(v',h')$};
  \node (w_h_tilde) at (-2, -0.1) {$(w,\widetilde{h})$};

  \draw[-{Latex[length=1.5mm, width=1.5mm]}] (v_h) -- (w_h_tilde) node[midway, right=1pt] {$P$};

  \coordinate (z1) at (0.6, -0.1);
  \coordinate (z2) at (0.1, -0.4);
  \coordinate (z3) at (1.5, -0.4);

  \draw[-{Latex[length=1.5mm, width=1.5mm]}] (w_h_tilde.east) -- (z1);
  \draw[-{Latex[length=1.5mm, width=1.5mm]}] (z1) -- (z2);
  \draw[-{Latex[length=1.5mm, width=1.5mm]}] (z2) -- (z3);
  \draw[-{Latex[length=1.5mm, width=1.5mm]}] (z3) -- (v_h_prime.south);

  \node[anchor=east] at (-2.5, 1.25) {$\exists m \quad hm=\widetilde{h}$};
  \node[anchor=west] at (2.0, 1.25) {$\exists m' \quad \widetilde{h}m'=h'$};

\end{tikzpicture}
\caption{Path re-visiting the $V \times H$ part}
\label{fig:path revisiting}
\end{figure}
        
    Therefore, $h~\mathcal{R}~\widetilde{h}$ and so, as the $\mathcal{R}$ classes of $M$ coincide with its $\mathcal{H}$ classes (from Lemma~\ref{L-commutativity same as R-classes coinciding with H-classes lemma}), we get $h~\mathcal{H}~\widetilde{h}$. However, $h$ and $\widetilde{h}$ are from different $\mathcal{H}$-classes (i.e., $H$ and $\widetilde{H}$ respectively), a contradiction. This proves Lemma~\ref{no re-visit lemma}. 
\end{proof}

\begin{lemma}
\label{eulerian pieces lemma}
Let $H$ be any group (i.e., $\mathcal{H}$-class) of $M$. Then,  $G'[V\times H]$ is an Eulerian subgraph of $G'$.
\end{lemma}

\begin{proof}
It suffices to show that the in-degree of any vertex matches its out-degree in $G'[V\times H]$. Consider any $(v, h)\in V\times H$. Note that the only edges of $G$ due to which edges get added (in the product graph $G'$) at the vertex $(v,h)\in V\times H$ are the ones incident to the vertex $v$ in $G$. Also, for any two such distinct edges (say, $\{v,w\}$ and $\{v,w'\}$, where $w$ and $w'$ are distinct neighbors of $v$ in $G$), the corresponding added edges at  $(v,h)$  would be such that their other endpoints would be of the form  $(w,\_)$ and $(w',\_)$ respectively; that is, the first component of these endpoints differs. So, $\{v,w\}$ and $\{v,w'\}$ cannot be the cause of the same edge incident at $(v,h)$. Thus, it now suffices to individually argue, for each edge $e=\{v,w\}\in E$ (say, with label $\lambda\in M$) incident to $v$ in $G$, that the numbers of in-edges and out-edges added (in $G'[V\times H]$)  at vertex $(v,h)$ due to edge $e$ are equal. 


\begin{description}
    \item[\textbf{Case 1}: $h\lambda \in H$:] 
    Then, due to $e$, there's clearly exactly one out-edge added to $(v,h)$ in $G'[V\times H]$; namely, from $(v,h)$ to $(w, h\lambda)$. Also, due to $e$, there's at least one in-edge added to $(v,h)$ in $G'[V\times H]$; namely, from $(w, h\cdot (h\lambda)^{-1}\cdot h)$ to $(v,h)$ respectively. Now, let us argue that this is the only in-edge added to $(v,h)$ in $G'[V\times H]$ due to $e$. For contradiction, suppose there exists $\mu\neq h\cdot (h\lambda)^{-1}\cdot h$ in $H$ such that there's an edge from $(w, \mu)$ to $(v,h)$ added due to $e$ in $G'[V\times H]$. Then, $\mu\cdot \lambda= h$. Writing $\mu$ as $\mu\cdot h^{-1}\cdot h$, we get $(\mu\cdot h^{-1})\cdot (h\cdot \lambda) = h$. Post-multiplying both sides by $(h\cdot \lambda)^{-1}\cdot h$, we get $\mu = h\cdot (h\lambda)^{-1}\cdot h$, a contradiction.
    \item[\textbf{Case 2}: $h\lambda \not\in H$:]
    Then, due to $e$, there's clearly no out-edge added to $(v,h)$ in $G'[V\times H]$. Let's argue that due to $e$, no in-edge is added either to $(v,h)$ in $G'[V\times H]$. For contradiction, assume that due to $e$, there's an edge from $(w, \mu)$ to $(v,h)$ added in $G'[V\times H]$. Then, $\mu\cdot \lambda = h$. Pre-multiplying both sides by $h\cdot \mu^{-1}$ gives $h\cdot \lambda = h\cdot \mu^{-1}\cdot h$, which $\in H$, a contradiction.
\end{description}    
    Thus, in both cases, the contributions to in-degree and out-degree of $(v,h)$ in $G'[V\times H]$ are equal.
    This proves Lemma~\ref{eulerian pieces lemma}. 
\end{proof}

\noindent{\em Proof of Theorem~\ref{Logspace algo when R-classes = H-classes} using Lemma~\ref{no re-visit lemma} and~\ref{eulerian pieces lemma}:}
Now, the graph $G$ has an $s$-to-$t$ walk of yield $\in F$  if and only if its product graph $G'$ has an $(s, id)$ to $(t,m)$ path for some $m\in F$, where $id$ denotes the monoid identity of $M$. Due to Lemma~\ref{no re-visit lemma}, this path (if it exists) must be such that for every $\mathcal{H}$-class $H$, it either never visits the part $V\times H$, or enters/leaves the part $V\times H$ only once. Guess (i.e., iterate over all possible choices of) the $\mathcal{H}$-classes $H$ whose corresponding $V\times H$ parts are visited, and also guess the edges involved in the jumps across these parts. As the number of $\mathcal{H}$-classes is $\mathcal{O}(1)$, there are only $n^{\mathcal{O}(1)}$ guesses (and so, the counter needed to iterate over them uses only $\mathcal{O}(\log n)$ space). For each such guess, the task now remains to check for each visited part $V\times H$, whether there is a path in $G'[V\times H]$ from the vertex at which the guessed edge entering into $V\times H$ lands, to the vertex at which the guessed edge exiting $V\times H$ departs from; this can be done using the logspace algorithm for Directed Reachability on Eulerian graphs because $G'[V\times H]$ is Eulerian, as shown in Lemma~\ref{eulerian pieces lemma}. This proves Theorem~\ref{Logspace algo when R-classes = H-classes}. 
\end{proof}

\begin{remark}
\label{R classes coincide with H classes implies L commutativity remark}
From Lemma~\ref{L-commutativity same as R-classes coinciding with H-classes lemma}, we know that for any union-of-groups monoid $M$, $\calL$-commutativity implies coincidence of $\mathcal{R}$-classes with $\mathcal{H}$-classes. In fact, the converse is also true. That is, if $\mathcal{R}$-classes of $M$ coincide with its $\mathcal{H}$-classes, then we can show that $M$ is $\calL$-commutative as follows: Suppose this is not the case. Then, there exist $a, b\in M$ such that $ab$ and $ba$ belong to different $\calL$-classes (say $L_1$ and $L_2$) respectively.  Let $H_1\subseteq L_1$ and $H_2\subseteq L_2$ denote the $\mathcal{H}$-classes containing $ab$ and $ba$ respectively. 
Note that $((ab)^{-1}\cdot a) \cdot (bab) =  ab$ and thus, $ab \leq _{\calL} bab$. Also, as $b\cdot ab = bab$, we have $bab \leq_{\calL} ab$. Thus, $bab$ belongs to the same $\calL$ class as $ab$ (i.e., $L_1$). Next, as $(bab)\cdot (a \cdot (ba)^{-1}) =  ba$, we get $ba \leq _{\mathcal{R}} bab$. Also, as $ba\cdot b = bab$, we get $bab \leq_{\mathcal{R}} ba$. Therefore, $bab$ and $ba$ belong to the same $\mathcal{R}$-class. So, as $\mathcal{R}$-classes coincide with $\mathcal{H}$-classes, both $bab$ and $ba$ belong to the same $\mathcal{H}$-class (i.e., $H_2$). However, $bab \in L_1$ is inconsistent with $bab\in H_2$, a contradiction.
\end{remark}

\begin{example}
Consider the full transformation monoid of $\{1,2\}$, i.e., the monoid consisting of all maps from $\{1,2\}$ to itself with function composition as the monoid operation. This union-of-groups monoid is $\calL$-commutative because it has the following same set of $\mathcal{R}$-classes and $\mathcal{H}$-classes: 
\begin{equation*}
\bigg\{\Big\{a = \left(\begin{array}{c}1\mapsto 2\\ 2\mapsto 1\end{array}\right), a^2 = \left(\begin{array}{c}1\mapsto 1\\ 2\mapsto 2\end{array}\right)\Big\}, \Big\{e= \left(\begin{array}{c}1 \mapsto 1\\ 2 \mapsto 1\end{array}\right)\Big\}, \Big\{ae = \left(\begin{array}{c}1 \mapsto 2\\ 2 \mapsto 2\end{array}\right)\Big\}\bigg\}.
\end{equation*}
\end{example}

\begin{remark}
\label{Clifford-ness implies L-commutativity}
It is known that Clifford monoids are both $\calL$-commutative and $\calR$-commutative (see Theorem 3.13 in \cite{mary2014classes}). We can show that any Clifford monoid $M$ is $\calL$-commutative (and by a similar argument, also $\calR$-commutative) as follows: Suppose this is not the case. Then, as per Remark~\ref{R classes coincide with H classes implies L commutativity remark}, $M$ has an $\mathcal{R}$-class that contains two of its distinct $\mathcal{H}$-classes $H_1$ and $H_2$, say with idempotent elements (same as group identities) $e_1$ and $e_2$ respectively. As $e_1$ and $e_2$ belong to the same $\mathcal{R}$-class, there exist $x, y\in M$ such that $e_1\cdot x = e_2$ and $e_2\cdot y = e_1$. Since the idempotent elements $e_1$ and $e_2$ commute with all elements of $M$ (because $M$ is a Clifford monoid), we can rewrite these as $x\cdot e_1 = e_2$ and $y\cdot e_2 = e_1$. So, apart from being $\mathcal{R}$-related, $e_1$ and $e_2$ are also $\calL$-related (and thus, $\mathcal{H}$-related) to each other, a contradiction.
\end{remark}
\vspace{-5mm}
\subsubsection{When $\mathcal{R}$-classes other than Unit Group are Right Ideals}
\label{subsubsec:When R-classes other than Unit Group are Right Ideals}
We now handle some interesting special cases where $\mathcal{R}$-classes other than the unit group are right ideals and the accepting subset $F$ is any of these $\cal{R}$-classes.

\firstrightidealalgo*


\begin{proof} Suppose that the accepting set $F$ is any $\mathcal{R}$-class (say $R$) other than $\tilde{G}$. We prove the following:
\begin{claim}
There is an $s$-$t$ walk of yield $\in R$ if and only if there exists an edge $e$ with label $\in M\setminus \tilde{G}$ such that (a) there is a walk $P_1$ of yield $\in \{g\in \tilde{G}~|~g\cdot \phi(e)\in R\}$ from $s$ to one endpoint of $e$, and (b) there is a path $P_2$ (of any yield) from the other endpoint of $e$ to $t$.
 \end{claim}

\begin{proof}
 $(\Leftarrow)$: Suppose that there is an edge $e$ with label $\in M\setminus \tilde{G}$ such that there is a walk $P_1$ of yield $g\in \tilde{G}$ satisfying $g\cdot \phi(e)\in R$ from $s$ to one endpoint of $e$, and there is a path $P_2$ (of any yield) from the other endpoint of $e$ to $t$. Then, concatenating $P_1$, $e$ and $P_2$ gives an $s$-$t$ walk whose yield is $g\cdot \phi(e)\cdot \phi(P_2)$, which is in $R$ because $g\cdot \phi(e)\in R$ and right-multiplying it with any element of $M$ always gives an element of $R$ (due to $R$ being a right ideal). 

$(\Rightarrow)$: Suppose that there is an $s$-$t$ walk $P$ of yield $\in R$. Note that $P$ must have at least one edge with label $\in M\setminus \tilde{G}$ (as any product of elements from just $\tilde{G}$ can never land outside $\tilde{G}$). Let $e$ be the first edge in $P$ which is labelled by an element of $M\setminus \tilde{G}$; that is, all edges before $e$ in $P$ are labelled by elements of $\tilde{G}$. Let $P_1$ and $P_2$ denote the subwalks of $P$ that appear before edge $e$ and after edge $e$, respectively. We have $\phi(P_1)\cdot \phi(e)\cdot \phi(P_2)\in R$. It suffices to show that $\phi(P_1)\cdot \phi(e)\in R$. For the sake of contradiction, assume that $\phi(P_1)\cdot \phi(e)\notin R$. 

\begin{flushleft}
\textbf{Case 1}: $\phi(P_1)\cdot \phi(e) \in \tilde{G}$.
\end{flushleft}
\vspace{-0.1 cm}
As $P_1$'s edges' labels are only from $\tilde{G}$, we have $\phi(P_1)\in \tilde{G}$. So, $(\phi(P_1))^{-1}\cdot \big(\phi(P_1)\cdot \phi(e)\big) \in \tilde{G}$ and thus, $\phi(e)\in \tilde{G}$, which is not possible as $\phi(e)\in M\setminus \tilde{G}$, a contradiction.

\begin{flushleft} 
\textbf{Case 2}: $\phi(P_1)\cdot \phi(e) \in R'$ for some $\mathcal{R}$-class $R'$ different from $R$ and $\tilde{G}$.
\end{flushleft}
\vspace{-0.1 cm}
 Then, $\big(\phi(P_1)\cdot \phi(e) \big) \cdot \phi (P_2)\in R'$ (due to $R'$ being a right ideal), but then it contradicts the fact that $\phi(P_1)\cdot \phi(e) \cdot \phi (P_2)\in R$. 
\end{proof}
Now, we get the desired log-space algorithm by simply iterating over all edges $e$ with labels from $M\setminus \tilde{G}$, and then checking whether conditions (a) and (b) (as stated in the claim above) hold. Note that (a) can be checked in logspace using our earlier analysis of the case when the accepting set is a subset of the unit group $\tilde{G}$ (i.e., see Theorem~\ref{thm: logspace algo for accepting identity for any finite monoid}),
and (b) can be checked using the usual logspace algorithm for the Reachability problem on undirected unlabelled graphs by \cite{Rei04}. This proves Theorem~\ref{R-classes right ideal thm}. 
\end{proof}

\subsection{Dichotomy for $\BA_2$ and $\U$ Monoids}
\label{subsec: Dichotomy for BA_2 and U with Different Accepting Sets}
While we are unable to handle all aperiodic monoids, we prove the following theorem for the two special aperiodic monoids $\BA_2$ and $\U$ 
(refer to table \ref{tab:BA_2U} for the Cayley tables). Recall that any non-commutative aperiodic monoid that is not in the class $\DA$ is divided by either $\BA_2$ or $\U$.

\begin{remark}
    \label{remark:BA2-U}
    When the monoid $M$ is $\BA_2$ (or $\U$ resp.), $\reachmf$ is in deterministic logspace with $F = M \setminus \{id\}$, where $id$ is the identity element of $M$. This can be done by checking if there is a walk from $s$ to $t$ in $G$ that uses at least one edge not labelled with the identity of $M$. Guess (iterate over all possible choices of) an edge $\{u,v\}$ in the desired walk that is not labelled with $id$. Then use the undirected reachability algorithm \cite{Rei04} to check if there is a path from $s$ to $u$ and $v$ to $t$ (of any yield). 
\end{remark}

\thmBAUsubset*

\begin{proof}  
    Fix the monoid $M =\BA_2$ (or $\U$ respectively) and the accepting set $F \subseteq M$. If $F=\{id\}$ or $M\setminus \{id\}$ or $M$ (here, $id = 1$), $\reachmf$ is logspace solvable (see Theorem \ref{thm: logspace algo for accepting identity for any finite monoid} and Remark \ref{remark:BA2-U}). Henceforth, assume that $F\neq \{1\},  M\setminus \{1\}$ or $M$.
    We will first consider the case where the accepting set $F$ is a singleton, and then proceed to other subsets. When $|F|=1$,  Theorem \ref{thm:BA2Usubset} says that $\reachmf$ is $\NL$-hard for $F=\{\alpha\}, \{\beta\}, \{\alpha\beta\}, \{\beta\alpha\}$ and logspace solvable when $F=\{0\}$.
    \\[-8mm]
   \subsubsection*{Case 1: $|F| = 1$.}$\reachmf$ is shown\footnote{We note that the proof for $\NL$-hardness for the case $\U$ described in \cite{RSS19} (Theorem 5.4), which is also used in~\cite{KMS26}, has a flaw (which was pointed out by an anonymous reviewer). Nevertheless, we note that the revised gadget in this version fixes the flaw and hence Theorem 5.4 in \cite{RSS19} holds true as it is. This version also revises the corresponding theorem in \cite{KMS26}.}
   to be $\NL$-hard for $M = \BA_2$ with $F=\{\alpha \beta\}$ \cite{KSS14,RSS19} and for $M = \U$ with $F = \{\alpha\}$ \cite{RSS19}. Based on this, we now show the complexity of $\reachmf$ for other singleton-accepting sets over $M$. 

    Case when $M=\BA_2$ with $F=\{\beta \alpha\}, \{\beta\}, \{\alpha\},\{0\}$ are shown in Lemma~\ref{thm:BA2beta_alpha}, Lemma~\ref{thm:BA2Ubeta}, Lemma~\ref{thm:BA2alpha}, and Lemma~\ref{lem:BA2_0} respectively. 
    Case when $M=\U$ with $F=\{\beta \alpha\},  \{\alpha\beta\}, \{\beta\},\{0\}$ are shown in Lemma~\ref{thm:BA2beta_alpha}, Lemma~\ref{thm:Ualpha_beta}, Lemma~\ref{thm:BA2Ubeta}, and Lemma~\ref{lem:U0} respectively.

    \begin{lemma}
    \label{thm:BA2beta_alpha}
        $\reachmf$ is $\NL$-hard for $M = \BA_2 (or ~ \U)$ with $F=\{\beta \alpha\}$.
    \end{lemma}

    The proof of Lemma \ref{thm:BA2beta_alpha} follows from Theorem  \ref{idempotent hardness with certain conditions}, where $\mu = \beta \alpha$ is an idempotent element, with factors $a = \beta$ and $b = \alpha$; case (ii) and case (iii) of Theorem \ref{idempotent hardness with certain conditions}
 is applicable for $\BA_2$ and $\U$ respectively.    

        \begin{lemma}
    \label{thm:Ualpha_beta}
        $\reachmf$ is $\NL$-hard for $M = \U$ with $F=\{\alpha\beta\}$.
    \end{lemma}
    
    The proof of Lemma~\ref{thm:Ualpha_beta} follows from case (i) of Theorem  \ref{idempotent hardness with certain conditions}, where $\mu =  \alpha \beta$ is an idempotent element, with factors $a = \alpha$ and $b = \beta$.

    \begin{lemma}
     \label{thm:BA2Ubeta}
        $\reachmf$ is $\NL$-hard for $M = \BA_2$ (or resp. $\U$ ) with $F=\{\beta\}$.
    \end{lemma}
    \begin{proof}
        We outline the proof idea since it is similar to the proof of Theorem \ref{idempotent hardness with certain conditions}
         with slight modifications. Fix the monoid $M$ to be $\BA_2$ (or resp. $\U$ ) and the accepting set $F=\{\beta\}$. The reduction is from a \textsc{Reach} instance $(G,s,t)$ to a $\reachmf$ instance $(G'(V',E'),\phi: E' \rightarrow \{\alpha,\beta\},s',t'')$. 
         The construction of $G'$ is similar to case (ii) of Theorem \ref{idempotent hardness with certain conditions}, where $\mu = \beta \alpha$ is the idempotent element, with factors $a = \beta$ and $b = \alpha$ and  $a^2=0=b^2$.
        Further, we introduce a new vertex $t''$ and add an edge from $t'$ (since there is an edge from $t$ to a new vertex $t'$ added with label $b$) to $t''$ with label $a=\beta$. 
        The correctness argument is similar to Theorem \ref{idempotent hardness with certain conditions}. 
    \end{proof} 

    \begin{lemma}
     \label{thm:BA2alpha}
        $\reachmf$ is $\NL$-hard for $M = \BA_2$ with $F=\{\alpha\}$.
    \end{lemma}
    
    \noindent The argument of the above lemma is similar to the previous lemmas.

    For $F=\{0\}$, we give a logspace algorithm\footnote{We remark that in conference version of this paper \cite{KMS26}, we had erroneously claimed that $\reachmf$ is $\NL$-hard for $M=\BA_2$ and $M=\U$ with $F=\{0\}$. Anonymous reviewers described a logspace algorithm for $M=\BA_2$ with $F=\{0\}$. We present this algorithm as Lemma \ref{lem:BA2_0}. Further, with a slight modification of this algorithm, we get a logspace algorithm that works for $M=\U$ as well; we present it in Lemma \ref{lem:U0}.} for $\reachmf$ for $M=\BA_2$ and $M=\U$ in Lemma \ref{lem:BA2_0} and Lemma \ref{lem:U0}, respectively. 

\begin{lemma}
     \label{lem:BA2_0}
        $\reachmf$ is in \ComplexityFont{L} for $M = \BA_2$ with $F=\{0\}$.
    \end{lemma}
    \begin{proof}
        It is sufficient to show the following claim. 

        \begin{claim}
Let $M=\BA_2$. An $s$--$t$ walk of yield $0$ exists if and only if $s$ and $t$
are connected and the connected component of $s$ contains an edge labelled by
$0$, $\alpha$ or $\beta$, or contains an edge labelled $\alpha\beta$ and an
edge labelled $\beta\alpha$.
\end{claim}
\begin{proof}
($\Leftarrow$) If the component has an edge labelled $0$, consider a walk with this edge from $s$ to $t$, and the yield becomes $0$ since it is an absorber. If it has an edge labelled $\alpha$ or $\beta$, taking a back-and-forth walk across this edge makes the yield $0$, since $\alpha^2=0$ and $\beta^2=0$. Else, all labels are in
$\{1,\alpha\beta,\beta\alpha\}$, and the component has an edge of each non-identity label. Walk from $s$ across an $\alpha\beta$-edge, then across a $\beta\alpha$-edge, then to $t$. The resulting word over $\{\alpha\beta,\beta\alpha\}$ contains both letters. So after deleting the identities, it contains two adjacent different letters, and $\alpha\beta\cdot\beta\alpha=\beta\alpha\cdot\alpha\beta=0$. 

($\Rightarrow$) By contrapositive, if all labels lie in $\{1,\alpha\beta\}$ or $\{1,\beta\alpha\}$, every yield is an idempotent $\alpha\beta$ or $\beta\alpha$ or $1$ (which is nonzero). 
\end{proof}
All the conditions in the claim can be checked in logspace using \cite{Rei04}.
    \end{proof}
    \vspace{-6mm}
\vspace{1 mm}
\begin{lemma}
\label{lem:U0}
$\reachmf$ is in $\mathsf{L}$ for $M=U$ with $F=\{0\}$.
\end{lemma}

\begin{proof}
Compute the connected component (say $C$) containing $s$ using \cite{Rei04}. If $C$ has least one edge (say with endpoints $u$ and $v$) labelled $\beta$ or $0$, then $C$ has an $s$-$t$ walk of yield $0$ obtained by concatenating $s$-$u$ walk, followed by back-and-forth traversal of edge $\{u,v\}$, followed by $v$-$t$ walk; so, we return YES. Next, we are in the case when all edges in $C$ have labels from $\{1,\alpha, \alpha\beta, \beta\alpha\}$.

Check whether $C$ has two edges (say $e_1$ and $e_2$) with labels $\alpha\beta$ and $\beta\alpha$ respectively such that $C$ has a path (say $P$) consisting of only $1$-labelled edges from some endpoint (say $v_1$) of $e_1$ to some endpoint (say $v_2$) of $e_2$; let $u_1$ and $u_2$ denote the other endpoints of $e_1$ and $e_2$ respectively. This can be checked by guessing $e_1, e_2, v_1, v_2$ (which can be done using a logspace counter as there are polynomially many such guesses), followed by using \cite{Rei04} to check whether there is a path from $v_1$ to $v_2$ using $1$-labelled edges of $C$. If yes, we return YES; otherwise, we return NO. Let us argue that this output so returned is YES if and only if $C$ has an $s$-$t$ walk of yield $0$.

$(\Rightarrow)$ First, suppose that the above output is YES. Then, $C$ has an $s$-$t$ walk of yield $0$ obtained by concatenating $s$-$u_1$ path, followed by edge $\{u_1, v_1\}$, followed by $v_1$-$v_2$ path $P$ of $1$'s, followed  by edge $\{v_2, u_2\}$, followed by $u_2$-$t$ path. Let length$(P) = \ell$. Since the $u_1$-$u_2$ subwalk of this walk has yield $(\alpha\beta)\cdot 1^{\ell}\cdot (\beta\alpha) = 0$, this forces the yield of the entire walk to be the absorber $0$, as desired.

$(\Leftarrow)$
Next, suppose that $C$ has an $s$-$t$ walk of yield $0$. To show that the output is YES, it suffices to prove that this walk has two edges $e_1$ and $e_2$ labelled $\alpha\beta$ and $\beta\alpha$ respectively, such that $e_1$ appears before $e_2$ in the walk and the subwalk between these two edges consists of only $1$-labelled edges. This walk naturally gives a word over $\{1,\alpha, \alpha\beta, \beta\alpha\}$ which evaluates to $0$. Also, skipping all $1$'s in this word gives a new word over $\{\alpha, \alpha\beta, \beta\alpha\}$ which evaluates to $0$. To complete the proof, it suffices to argue that this word must have consecutive positions containing $\alpha\beta$ followed by  $\beta\alpha$.  
To that end, let us prove (using induction on the length of $x$) that for any word $x$ over $\{\alpha, \alpha\beta, \beta\alpha\}$ such that $x$ evaluates to $0$, $x$ must have two consecutive positions containing $\alpha\beta$ followed by $\beta\alpha$. The claim clearly holds for the base case (i.e., $|x|=2$). As the induction hypothesis, assume that it holds for all words of length $\leq n-1$, for some $n\geq 3$. As the induction step, let us argue that it also holds for all $n$-length words $x$. Consider an $n$-length word $x$ over $\{\alpha, \alpha\beta,\beta\alpha\}$ that evaluates to $0$. 

\begin{description}
        \item[Case 1:] $x[1]=\alpha$.\\First, assume that $x[2]\in \{\alpha, \beta\alpha\}$. Then, observe that the $(n-1)$-length word $x':=x[1]x[3]x[4]\ldots$ (i.e., obtained from $x$ by omitting its second letter) has the same yield as $x$ because $\alpha^2=\alpha$ and $\alpha\cdot (\beta\alpha)=\alpha$. Using the induction hypothesis, $x'$ has two consecutive positions containing $\alpha\beta$ followed by $\beta\alpha$; also, these two positions are not the first two (i.e., $x[1],x[3]$; recall that $x[2]$ is omitted in $x'$) because $x[1]=\alpha\neq \alpha\beta$. Thus, these two positions must be from the portion $x[3]x[4]\ldots$ and so, these are also consecutive in $x$.

Next, assume that $x[2]=\alpha\beta$. Then, as $\alpha\cdot (\alpha\beta) = \alpha\beta$, the word $x':=x[2]x[3]x[4]\ldots$ (i.e., obtained from $x$ by omitting its first letter) has the same yield as $x$. Again, using the induction hypothesis, $x'$ (and so, also $x$) has two consecutive positions containing $\alpha\beta$ followed by $\beta\alpha$, as desired.

\item[Case 2:] $x[1]=\alpha\beta$.\\
If $x[2]=\beta\alpha$, then $x[1]$ and $x[2]$ are the desired consecutive positions. Next, assume that $x[2]\in \{\alpha,\alpha\beta\}$. Then, as $(\alpha\beta)^2=\alpha\beta$ and $(\alpha\beta)\cdot \alpha = \alpha$, the word $x':= x[2]x[3]x[4]\ldots$ has the same yield as $x$. Again, using the induction hypothesis, $x'$ (and so, also $x$) has two consecutive positions containing $\alpha\beta$ followed by $\beta\alpha$, as desired.

\item[Case 3:] $x[1]=\beta\alpha$.\\
First, assume that $x[2]\in \{\alpha, \beta\alpha\}$. Then, as $(\beta\alpha)\cdot \alpha=\beta\alpha$ and $(\beta\alpha)^2=\beta\alpha$, the word $x':= x[1]x[3]x[4]\ldots$ has the same yield as $x$. Again, using the induction hypothesis, $x'$ (and so, also $x$) has two consecutive positions containing $\alpha\beta$ followed by $\beta\alpha$, as desired.

Next, assume that $x[2]=\alpha\beta$. If $x[3]=\beta\alpha$, $x[2]$ and $x[3]$ are the desired consecutive positions. So now, assume that $x[3]\in \{\alpha, \alpha\beta\}$. Then, as $(\alpha\beta)\cdot \alpha=\alpha$ and $(\alpha\beta)^2=\alpha\beta$, the word $x':=x[1]x[3]x[4]\ldots$ has the same yield as $x$. Again, using the induction hypothesis, $x'$ (and so, also $x$) has two consecutive positions containing $\alpha\beta$ followed by $\beta\alpha$, as desired.
\end{description}

This concludes the proof of Lemma \ref{lem:U0}.
\end{proof}
        
This completes the analysis of space complexity of $\reachmf$ for $M=BA_2$ and $M=U$ for all singleton accepting sets $F\subseteq M$. In the next case, we consider  $F\subseteq M$ of size $\geq 2$.       \subsubsection*{Case 2: $|F| \geq 2$.}

When $F$ contains at least one of $\alpha, \beta, \alpha\beta, \beta\alpha$ but not $0$, then $\reachmf$ for $M=\BA_2$ and $\U$ is \NL-hard using the same reductions used to show \NL-hardness for accepting sets $\{\alpha\}$, $\{\beta\},\{\alpha\beta\}, \{\beta\alpha\}$ in the previous case; this is because in the reverse direction in the proof of Theorem \ref{idempotent hardness with certain conditions}, a contradiction was arrived at by having yield $0$, which remains a contradiction here because $0\not\in F$. 

Next, it remains to show that $\reachmf$ is logspace solvable for $M=\BA_2$ and $\U$ when the accepting set $F$ contains $0$. We first present the argument for $M=\U$ below. To that end, it suffices to show that $\reachmf$ is logspace solvable for $M=\U$ when $F=\{0,\beta\}, \{0,\alpha\}, \{0,\alpha\beta\}$ or $\{0,\beta\alpha\}$ (because any set $F$ of size $\geq 2$ containing $0$ is a union of some of these two-sized sets, possibly along with $1$); we argue so in Lemma \ref{lem:U_0_and_beta}, Lemma \ref{lem:U_0_and_alpha}, Lemma \ref{lem:U_0_and_alphabeta} and Lemma \ref{lem:U_0_and_betaalpha} respectively.

\begin{lemma}
\label{lem:U_0_and_beta}
$\reachmf$ is logspace solvable when $M=\U$ with $F=\{0,\beta\}$.
\end{lemma}
\begin{proof}
Compute the connected component (say $C$) containing $s$ in $G$. First, 
check whether $C$ has a $s$-$t$ walk of yield $0$ using Lemma \ref{lem:U0}; if yes, return YES. Henceforth, assume that $C$ has no $s$-$t$ walk of yield $0$. Now, all edges in $C$ have labels from $\{1,\alpha,\alpha\beta,\beta\alpha\}$ (as otherwise, the algorithm in Lemma \ref{lem:U0} would have already returned YES). We check whether there are edges $e_1$ and $e_2$ in $C$ labelled $\beta\alpha$ and $\alpha\beta$ respectively such that (i) there is a path (say $P_1$) of only $1$-labelled edges from $s$ to one endpoint (say $u$; let $u'$ denotes the other endpoint) of $e_1$, and (ii) there is a path (say $P_2$) of only $1$-labelled edges from one endpoint (say $v$; let $v'$ denote the other endpoint) of $e_2$ to $t$. If yes, return YES; otherwise, return NO. We argue that this output is YES iff $C$ has an $s$-$t$ walk of yield $\beta$. 

First, suppose that the above output is YES. Then, consider the $s$-$t$ walk in $C$ obtained by concatenating $P_1$, followed by $e_1$, followed by any path in $C$ from $u'$ to $v'$, followed by $e_2$, followed by $P_2$. The yield of this walk belongs to $(\beta\alpha)\cdot \U \cdot (\alpha\beta) = \{0,\beta\}$. Also, it cannot be $0$ (as otherwise, Lemma \ref{lem:U0} above would have already returned YES). Therefore, its yield is $\beta$, as desired. 

Next, suppose that $C$ has an $s$-$t$ walk of yield $\beta$. Let $e_1$ and $e_2$ denote the first edge labelled $\neq 1$ and the last edge labelled $\neq 1$ in this walk. Now, to show that the output returned above is YES, it suffices to argue that $e_1$ and $e_2$ are labelled $\beta\alpha$ and $\alpha\beta$ respectively. If $e_1$'s label is not $\beta\alpha$, then its label is from $\{\alpha,\alpha\beta\}$ and so, the walk's yield belongs to $\{\alpha,\alpha\beta\}\cdot \U = \alpha\cdot \U = \{0,\alpha,\alpha\beta\}$, which does not contain $\beta$. Also, if $e_2$'s label is not $\alpha\beta$, then its label is from $\{\alpha,\beta\alpha\}$ and so, the walk's yield belongs to $\U\cdot \{\alpha,\beta\alpha\} = \U\cdot\alpha = \{0,\alpha,\beta\alpha\}$, which again does not contain $\beta$. Therefore, $e_1$ and $e_2$ must indeed be labelled $\beta\alpha$ and $\alpha\beta$ respectively, as desired. This concludes the proof of Lemma \ref{lem:U_0_and_beta}.
\end{proof}

\begin{lemma}
\label{lem:U_0_and_alpha}
$\reachmf$ is logspace solvable when $M=\U$ with $F=\{0,\alpha\}$.
\end{lemma}
\begin{proof}
As in the proof of Lemma \ref{lem:U_0_and_beta} above, we first use Lemma \ref{lem:U0} to return YES if the component $C$ containing $s$ has a $s$-$t$ walk of yield $0$. Henceforth, again, all edges in $C$  have labels from $\{1,\alpha,\alpha\beta,\beta\alpha\}$. We also check whether there is an edge $e$ (say with endpoints $x$ and $y$) labelled $\alpha$ such that there are $s$-$x$  and $y$-$t$ paths of only $1$-labelled edges in $C$. If yes, we return YES, as then $C$ has an $s$-$t$ walk of yield $\alpha$, namely the one obtained by concatenating this $s$-$x$ path, followed by edge $e$, followed by this $y$-$t$ path. So, for the remaining proof, we can safely assume that any $s$-$t$ walk of yield $\alpha$ (if it exists) in $C$ must have at least two distinct edges labelled $\neq 1$.

Next, we check whether there are two distinct edges $e_1$ and $e_2$ in $C$ such that (i) $e_1$ is labelled $\alpha$ or $\alpha\beta$ and $e_2$ is labelled $\alpha$ or $\beta\alpha$, (ii) there is a path (say $P_1$) of only $1$-labelled edges from $s$ to one endpoint (say $u$; let $u'$ denotes the other endpoint) of $e_1$, and (iii) there is a path (say $P_2$) of only $1$-labelled edges from one endpoint (say $v$; let $v'$ denote the other endpoint) of $e_2$ to $t$. If yes, return YES; otherwise, return NO. We argue that this output is YES iff $C$ has an $s$-$t$ walk of yield $\alpha$. 

First, suppose that the above output is YES. Then, consider the $s$-$t$ walk in $C$ obtained by concatenating $P_1$, followed by $e_1$, followed by any path in $C$ from $u'$ to $v'$, followed by $e_2$, followed by $P_2$. The yield of this walk belongs to $\alpha\cdot \U\cdot \alpha = \{\alpha,0\}$ or $\alpha\cdot \U\cdot (\beta\alpha) = \{\alpha,0\}$ or $(\alpha\beta)\cdot \U\cdot \alpha = \{\alpha,0\}$ or $(\alpha\beta)\cdot \U\cdot (\beta\alpha) = \{\alpha,0\}$. Also, it cannot be $0$ (as otherwise, Lemma \ref{lem:U0} above would have already returned YES). Therefore, its yield is $\alpha$, as desired. 

Next, suppose that $C$ has an $s$-$t$ walk of yield $\alpha$. Let $e_1$ and $e_2$ denote the first edge labelled $\neq 1$ and the last edge labelled $\neq 1$, respectively in this walk. As ensured by the check in the first paragraph, $e_1$ and $e_2$ are distinct. Now, to show that the output returned above is YES, it suffices to argue that $e_1$ is labelled $\alpha$ or $\alpha\beta$ and $e_2$ is labelled $\alpha$ or $\beta\alpha$. If $e_1$'s label is neither $\alpha$ nor $\alpha\beta$, then its label is $\beta\alpha$ and so, the walk's yield belongs to $(\beta\alpha)\cdot \U = \{0,\beta,\beta\alpha\}$, which does not contain $\alpha$. Also, if $e_2$'s label is neither $\alpha$ nor $\beta\alpha$, then its label is $\alpha\beta$ and so, the walk's yield belongs to $\U\cdot(\alpha\beta) = \{0,\beta,\alpha\beta\}$, which again does not contain $\alpha$. Therefore, $e_1$ (resp. $e_2$) must indeed be labelled $\alpha$ or $\alpha\beta$ (resp. $\beta\alpha$), as desired. This concludes the proof of Lemma \ref{lem:U_0_and_alpha}.
\end{proof}

\begin{lemma}
\label{lem:U_0_and_alphabeta}
$\reachmf$ is logspace solvable when $M=\U$ with $F=\{0,\alpha\beta\}$.
\end{lemma}
\begin{proof}
As in the proof of Lemma \ref{lem:U_0_and_beta} above, we first use Lemma \ref{lem:U0} to return YES if the component $C$ containing $s$ has a $s$-$t$ walk of yield $0$. Henceforth, again, all edges in $C$  have labels from $\{1,\alpha,\alpha\beta,\beta\alpha\}$. We also check whether there is an edge $e$ (say with endpoints $x$ and $y$) labelled $\alpha\beta$ such that there are $s$-$x$  and $y$-$t$ paths of only $1$-labelled edges in $C$. If yes, we return YES, as then $C$ has an $s$-$t$ walk of yield $\alpha\beta$, namely the one obtained by concatenating this $s$-$x$ path, followed by edge $e$, followed by this $y$-$t$ path. So, for the remaining proof, we can safely assume that any $s$-$t$ walk of yield $\alpha\beta$ (if it exists) in $C$ must have at least two distinct edges labelled $\neq 1$.

Next, we check whether there are two distinct edges $e_1$ and $e_2$ in $C$ such that i) $e_1$ is labelled $\alpha$ or $\alpha\beta$ and $e_2$ is labelled $\alpha\beta$, ii) there is a path (say $P_1$) of only $1$-labelled edges from $s$ to one endpoint (say $u$; let $u'$ denotes the other endpoint) of $e_1$, and iii) there is a path (say $P_2$) of only $1$-labelled edges from one endpoint (say $v$; let $v'$ denote the other endpoint) of $e_2$ to $t$. If yes, return YES; otherwise, return NO. We argue that this output is YES iff $C$ has an $s$-$t$ walk of yield $\alpha\beta$. 

First, suppose that the above output is YES. Then, consider the $s$-$t$ walk in $C$ obtained by concatenating $P_1$, followed by $e_1$, followed by any path in $C$ from $u'$ to $v'$, followed by $e_2$, followed by $P_2$. The yield of this walk belongs to $\alpha\cdot \U\cdot (\alpha\beta) = \{0,\alpha\beta\}$ or $(\alpha\beta)\cdot \U\cdot (\alpha\beta) = \{0,\alpha\beta\}$. Also, it cannot be $0$ (as otherwise, Lemma \ref{lem:U0} above would have already returned YES). Therefore, its yield is $\alpha\beta$, as desired. 

Next, suppose that $C$ has a $s$-$t$ walk of yield $\alpha\beta$. Let $e_1$ and $e_2$ denote the first edge labelled $\neq 1$ and the last edge labelled $\neq 1$, respectively in this walk. As ensured by the check in the first paragraph, $e_1$ and $e_2$ are distinct. Now, to show that the output returned above is YES, it suffices to argue that $e_1$ is labelled $\alpha$ or $\alpha\beta$ and $e_2$ is labelled  $\alpha\beta$. If $e_1$'s label is neither $\alpha$ nor $\alpha\beta$, then its label is $\beta\alpha$ and so, the walk's yield belongs to $(\beta\alpha)\cdot \U = \{0,\beta,\beta\alpha\}$, which does not contain $\alpha\beta$. Also, if $e_2$'s label is not $\alpha\beta$, then its label is from $\{\beta\alpha, \alpha\}$ and so, the walk's yield belongs to $\U\cdot(\beta\alpha) = \{0,\alpha,\beta\alpha\}$ or $\U\cdot\alpha = \{0,\alpha,\beta\alpha\}$, which again does not contain $\alpha\beta$. Therefore, $e_1$ must indeed be labelled $\alpha$ or $\alpha\beta$ and $e_2$ must indeed be labelled $\alpha\beta$, as desired. This concludes the proof of Lemma \ref{lem:U_0_and_alphabeta}.
\end{proof}

\begin{lemma}
\label{lem:U_0_and_betaalpha}
$\reachmf$ is logspace solvable when $M=\U$ with $F=\{0,\beta\alpha\}$.
\end{lemma}
\begin{proof}
As in the proof of Lemma \ref{lem:U_0_and_beta} above, we first use Lemma \ref{lem:U0} to return YES if the component $C$ containing $s$ has a $s$-$t$ walk of yield $0$. Henceforth, again, all edges in $C$  have labels from $\{1,\alpha,\alpha\beta,\beta\alpha\}$. We also check whether there is an edge $e$ (say with endpoints $x$ and $y$) labelled $\beta\alpha$ such that there are $s$-$x$  and $y$-$t$ paths of only $1$-labelled edges in $C$. If yes, we return YES, as then $C$ has an $s$-$t$ walk of yield $\beta\alpha$, namely the one obtained by concatenating this $s$-$x$ path, followed by edge $e$, followed by this $y$-$t$ path. So, for the remaining proof, we can safely assume that any $s$-$t$ walk of yield $\beta\alpha$ (if existent) in $C$ must have at least two distinct edges labelled $\neq 1$.

Next, we check whether there are two distinct edges $e_1$ and $e_2$ in $C$ such that i) $e_1$ is labelled $\beta\alpha$ and $e_2$ is labelled $\alpha$ or $\beta\alpha$, ii) there is a path (say $P_1$) of only $1$-labelled edges from $s$ to one endpoint (say $u$; let $u'$ denotes the other endpoint) of $e_1$, and iii) there is a path (say $P_2$) of only $1$-labelled edges from one endpoint (say $v$; let $v'$ denote the other endpoint) of $e_2$ to $t$. If yes, return YES; otherwise, return NO. We argue that this output is YES iff $C$ has a $s$-$t$ walk of yield $\beta\alpha$. 

First, suppose that the above output is YES. Then, consider the $s$-$t$ walk in $C$ obtained by concatenating $P_1$, followed by $e_1$, followed by any path in $C$ from $u'$ to $v'$, followed by $e_2$, followed by $P_2$. The yield of this walk belongs to $ (\beta\alpha)\cdot\U\cdot \alpha = \{0,\beta\alpha\}$ or $(\beta\alpha)\cdot\U\cdot(\beta\alpha) = \{0,\beta\alpha\}$. Also, it cannot be $0$ (as otherwise, Lemma \ref{lem:U0} above would have already returned YES). Therefore, its yield is $\beta\alpha$, as desired. 

Next, suppose that $C$ has a $s$-$t$ walk of yield $\beta\alpha$. Let $e_1$ and $e_2$ denote the first edge labelled $\neq 1$ and the last edge labelled $\neq 1$, respectively in this walk. As ensured by the check in the first paragraph, $e_1$ and $e_2$ are distinct. Now, to show that the output returned above is YES, it suffices to argue that $e_1$ is labelled $\beta\alpha$ and $e_2$ is labelled  $\alpha$ or $\beta\alpha$. If $e_1$'s label is not $\beta\alpha$, then its label is from $\{\alpha,\alpha\beta\}$ and so, the walk's yield belongs to $\alpha\cdot \U = \{0,\alpha,\alpha\beta\}$ or $(\alpha\beta)\cdot \U = \{0,\alpha,\alpha\beta\}$, which does not contain $\beta\alpha$. Also, if $e_2$'s label is neither  $\alpha$ nor $\beta\alpha$, then its label is $\alpha\beta$ and so, the walk's yield belongs to $\U\cdot(\alpha\beta) = \{0,\beta,\alpha\beta\}$, which again does not contain $\beta\alpha$. Therefore, $e_1$ must indeed be labelled $\beta\alpha$ and $e_2$ must indeed be labelled $\alpha$ or $\beta\alpha$, as desired. This concludes the proof of Lemma \ref{lem:U_0_and_betaalpha}.
\end{proof}
\end{proof}

Next, we show that $\reachmf$ is logspace solvable for $M=\BA_2$ when $F=\{0,\beta\}, \{0,\alpha\}, \{0,\alpha\beta\}$ or $\{0,\beta\alpha\}$, similar to $M = \U$.

\begin{lemma}
\label{lem:BA2_0_and_beta}
$\reachmf$ is logspace solvable when $M=\BA_2$ with $F=\{0,\beta\}$.
\end{lemma}

We remark that testing reachability in an undirected graph labelled with elements from $\BA_2$ with accepting set $F=\{0,\beta\}$ is the same as testing reachability with accepting set $F=\{0\}$. This is because an $s-t$ walk with yield $\beta$ must necessarily contain at least one edge labelled $\beta$. Taking an $s-t$ walk by going back-and-forth on this edge yields $0$ since $\beta^2=0$ in $\BA_2$. Hence, the proof of Lemma~\ref{lem:BA2_0_and_beta} follows from the proof of Lemma~\ref{lem:BA2_0}.

\begin{lemma}
\label{lem:BA2_0_and_alpha}
$\reachmf$ is logspace solvable when $M=\BA_2$ with $F=\{0,\alpha\}$.
\end{lemma}

An argument similar to Lemma~\ref{lem:BA2_0_and_beta} can be made for Lemma~\ref{lem:BA2_0_and_alpha} since $\alpha^2=0$ in $\BA_2$.

\begin{lemma}
\label{lem:BA2_0_and_alphabeta}
$\reachmf$ is logspace solvable when $M=\BA_2$ with $F=\{0,\alpha\beta\}$.
\end{lemma}
\begin{proof}
Compute the connected component (say $C$) containing $s$ in $G$. First, check whether $C$ has an $s-t$ walk of yield $0$ using Lemma \ref{lem:BA2_0}; if yes, return YES. Henceforth, assume that $C$ has no $s$-$t$ walk of yield $0$. Now, all edges in $C$ have labels from $\{1,\alpha\beta,\beta\alpha\}$ (as otherwise, the algorithm in Lemma \ref{lem:U0} would have already returned YES). 
Now, w.k.t. for the final yield to be $\alpha\beta$, there must be at least one edge labelled $\alpha\beta$. 
Since $\alpha\beta \cdot \beta \alpha = \beta\alpha \cdot \alpha\beta = 0$, any $s-t$ path cannot have a $\beta\alpha$ edge for the yield to be $\alpha\beta$. Hence, drop all the edges labelled $\beta\alpha$. Then use the undirected unlabelled reachability algorithm \cite{Rei04} to test reachability on this graph.
\end{proof}

\begin{lemma}
\label{lem:BA2_0_and_betaalpha}
$\reachmf$ is logspace solvable when $M=\BA_2$ with $F=\{0,\beta\alpha\}$.
\end{lemma}

An argument similar to Lemma~\ref{lem:BA2_0_and_alphabeta} can be made for Lemma~\ref{lem:BA2_0_and_betaalpha} since $\alpha^2=0=\beta^2$ in $\BA_2$.

\section{Discussion and Conclusion}
In this paper, we presented complexity-theoretic upper and lower bounds for  $\reachmf$ for various subclasses of monoids. The main technical ingredient is the study of the product graph's structure, which is established by exploiting the algebraic properties of the underlying monoid. The general goal of this work is to completely classify the difficulty of the labelled reachability problem based on the structure of the monoid from which the elements are used for labelling. While this target may be far away, one immediate question is whether one can establish a deterministic logspace algorithm for the reachability problem when the underlying monoid $M$ is a union-of-groups and the accepting subset $F$ is an arbitrary subset of $M$. 


\paragraph{Funding:} The first author's research is supported by PhD Fellowship from the Ministry of Education, Government of India. The second author's work is supported by the Institute Postdoctoral Fellowship at IIT Madras.

\paragraph{Acknowledgements:}
The authors thank the anonymous reviewers for their useful comments on a previous version of this paper.
We are grateful to the reviewer who pointed out a mistake in the earlier reduction in Theorem \ref{idempotent hardness with certain conditions} and a part of Theorem~\ref{thm:BA2Usubset}.  
We are also thankful to an anonymous reviewer of the conference version who pointed out a relaxation in the premise of Theorem \ref{idempotent hardness with certain conditions} and how that can be used to argue Theorem~\ref{thm: monoid outside DS hardness for some idempotent} by a proof using Rees matrix representation of a related semigroup (See section~\ref{sec:reviewers-proof}).

\bibliographystyle{alpha}
\bibliography{references.bib}

\end{document}